\documentclass{article}

\usepackage{times}
\usepackage{amssymb,amsmath,amsthm,comment}
\usepackage{mathrsfs}
\usepackage{epsfig}
\usepackage{color}
\usepackage{enumitem}

\def\B{\mathscr B}
\def\C{\mathbb C}
\def\d{\mathrm{d}}
\def\F{\mathscr F}
\def\G{\mathcal G}

\def\H{\mathcal H}

\def\K{\mathscr K}
\def\L{\mathcal L}

\def\R{\mathbb R}
\def\S{\mathbb S}
\def\SS{\mathscr S}

\def\Hrond{\mathscr H}
\def\HS{\mathfrak h}
\def\Pv{\mathrm{Pv}}
\def\EE{\mathscr E}
\def\Z{\mathbb Z}

\def\O{\mathcal{O}}

\def\dom{\mathcal D}
\def\lone{\mathsf{L}^{\:\!\!1}}

\def\ltwo{\mathsf{L}^{\:\!\!2}}

\def\e{\mathop{\mathrm{e}}\nolimits}

\def\supp{\mathop{\mathrm{supp}}\nolimits}

\def\slim{\mathop{\hbox{\rm s-}\lim}\nolimits}

\def\ulim{\mathop{\hbox{\rm u-}\lim}\nolimits}

\def\Ran{\mathop{\mathsf{Ran}}\nolimits}

\def\Span{\mathop{\mathrm{span}}\nolimits}
\def\index{\mathop{\mathrm{Index}}\nolimits}
\def\dim{\mathop{\mathrm{dim}}\nolimits}
\def\Tr{\mathop{\mathrm{Tr}}\nolimits}

\def\wind{\mathop{\mathrm{Wind}}\nolimits}
\def\vol{\mathop{\mathrm{Vol}}\nolimits}
\def\ess{{\mathrm{ess}}}
\def\disc{{\mathrm{disc}}}
\def\dim{{\mathrm{dim}}}
\def\cl{{\mathrm{cl}}}
\def\ads{{\mathrm{ads}}}

\newtheorem{Theorem}{Theorem}[section]
\newtheorem{Remark}[Theorem]{Remark}
\newtheorem{Lemma}[Theorem]{Lemma}

\newtheorem{Corollary}[Theorem]{Corollary}
\newtheorem{Proposition}[Theorem]{Proposition}

\begin{document}


\title{Levinson’s theorem for dissipative systems, \\
or how to count the asymptotically disappearing states}

\author{A. Alexander${}^1$, J. Faupin${}^2$, S. Richard${}^3$}

\date{\small}
\maketitle
\vspace{-1cm}

\begin{quote}
\begin{itemize}
\item[1] School of Mathematics and Applied Statistics, University of Wollongong,
Northfields Ave, \\ Wollongong, NSW 2522, Australia
\item[2] Institut Elie Cartan de Lorraine, Universit\'e de Lorraine, 57045 Metz Cedex 1,
France
\item[3] Institute for Liberal Arts and Sciences \& Graduate School of Mathematics, Nagoya University, Furo-cho, Chikusa-ku, Nagoya, 464-8601, Japan
\item[] E-mail:  angusa@uow.edu.au,  jeremy.faupin@univ-lorraine.fr, richard@math.nagoya-u.ac.jp
\end{itemize}
\end{quote}


\begin{abstract}
We consider dissipative Schr\"odinger operators of the form $H=-\Delta+V(x)$ on $\ltwo(\R^3)$, with $V(x)$ a complex, bounded and decaying potential with a non-positive imaginary part. We prove a topological version of Levinson's theorem corresponding to an index theorem for the discrete, complex spectrum of $H$.
\end{abstract}


\smallskip


\section{Introduction}

Every quantum system is dissipative, to some extent, in the sense that part of its energy may be irreversibly transferred to the environment. 
There are various ways to mathematically represent a dissipative quantum system. 
In the theory of open quantum systems, see for example \cite{Da76,Ex85}, the system of interest is coupled to a reservoir  
representing the environment. The latter can be modeled, for instance, by a bosonic field. 
A celebrated effective description of the evolution of dissipative quantum systems, 
obtained by first tracing out the degrees of freedom of the environment, is given by Lindblad's master equation \cite{Li76}. 
In this setting the mixed states of the system are associated to density matrices -- positive trace class operators of trace $1$ -- 
and the generator of the dynamics, the Lindbladian, acts on the Banach space of trace class operators. 
Another simple, phenomenological way to describe a particular dissipative quantum system, given by a neutron interacting with a fixed nucleus, 
was introduced in \cite{FePoWe54_01}. It is now called the nuclear optical model. Here the dynamics of the neutron is generated 
by a Schr\"odinger operator
\begin{equation}\label{eq:Hintro}
H=H_0+V
\end{equation}
on $\ltwo(\R^3)$, where $H_0=-\Delta$ corresponds to the free, kinetic energy of the neutron (in suitable units) 
and $V$ is a complex potential with non-positive imaginary part corresponding to the effective interaction 
between the neutron and the nucleus. 
A feature of the model is that it allows for a description of the phenomenon of capture of the neutron by the nucleus, 
leading to the formation of a compound nucleus. 
The form of the potential $V$ is deduced from the scattering data obtained from experiments. 
We refer to \cite{Ho71,Fe92} for extensions of the model to more complex situations. 
In principle, one can similarly represent the dynamics of a dissipative quantum system, 
which can be captured by another system, by a semi-group of contractions generated 
by a dissipative operator in the Hilbert space corresponding to the pure states of the system.

The spectral and scattering theory of dissipative operators in Hilbert spaces, 
in relation with the nuclear optical model, have been studied by many authors. 
We refer, among others, to the works of Martin \cite{Ma75}, Davies \cite{Da79,Da80_01} and Neidhardt \cite{Ne85} 
for general results in an abstract setting, Mochizuki \cite{Mo68}, Simon \cite{Si79}, Wang \cite{Wa11,Wa12}, Wang and Zhu \cite{WaZh14} 
for dissipative Sch\"odinger operators of the form \eqref{eq:Hintro}, and 
\cite{Fa21,FaFr18,FN} for more recent developments. 
Let us also mention \cite{FaFaFrSc17,FaFr18} for applications to the scattering theory of the Lindblad master equation, 
and \cite{FaFr23,Fr24} for extensions to non-dissipative operators still related to complex Schr\"odinger operators 
and to the nuclear optical model.

In dissipative quantum scattering theory, the wave operators are defined by the time-dependent expressions
\begin{equation}\label{eq:waveintro}
W_-:=\slim_{t\to \infty}\e^{-itH} \e^{itH_0}\quad  \hbox{ and }\quad 
W_+:=\slim_{t\to \infty} \e^{itH^*}\e^{-itH_0},
\end{equation}
where $H^*=H_0+\overline{V}$ stands for the adjoint of $H$, whenever these strong limits exist. 
In this case, the scattering operator is defined by the relation
\begin{equation}\label{eq:scatt_operator_intro}
S\equiv S(H,H_0):= W_+^* W_-.
\end{equation}
It will be recalled in the next section that these three operators are contractions. Since $S$ commutes with $H_0$, 
it can be diagonalized in the sense that, for a.e.~$\lambda>0$ and $\omega\in\S^2$,
$$
[\F_0 S \F_0^*\varphi](\lambda,\omega) =\big[S(\lambda)\varphi(\lambda)\big](\omega),
$$
where $\F_0:\H\to \ltwo(\R_+, \d \lambda; \HS)$ is a unitary map satisfying 
$$
\big[[\F_0 H_0 \F_0^*\varphi](\lambda)\big]( \omega) = \lambda [\varphi(\lambda)](\omega),
$$
for suitable $\varphi \in \ltwo(\R_+, \d \lambda; \HS)$, with $\HS:=\ltwo(\S^2)$. 
We refer again to the next section for more details. 
The scattering matrix $S(\lambda)$, as an operator acting on $\HS$, is also a contraction. 
In particular, contrary to the self-adjoint setting, $S(\lambda)$ may not be invertible. 
An energy $\lambda\in[0,\infty)$ such that $S(\lambda)$ is not invertible is called a spectral singularity. 
In our context, it corresponds to a real resonance for the Schr\"odinger operator $H$. 
In the sequel we shall assume that all $\lambda>0$ are regular points for $H$, 
i.e. that $S(\lambda)$ is always invertible.

As in the works \cite{FaFr18,FN}, we introduce two subspaces specific to dissipative systems, namely
\begin{align}\label{eq:def_Hads}
\H_\ads(H) & :=\big\{f\in \H\mid \lim_{t\to \infty}\|\e^{-itH} f\|=0\big\}, 
\quad \H_\ads(H^*) :=\big\{f\in \H\mid \lim_{t\to \infty}\|\e^{itH^*} f\|=0\big\}, 
\end{align}
and call them the spaces of \emph{asymptotically disappearing states} for $H$ and $H^*$, respectively. 
For the nuclear optical model, these spaces play an important role as the normalized vectors in $\H_\ads(H)$ 
correspond to the initial states of the neutron which lead to an absorption by the nucleus, 
asymptotically as $t\to\infty$. 
One can verify that $\H_\ads(H)$ contains the discrete spectral subspace $\H_\disc(H)$, defined as the closure of the vector space spanned by the generalized 
eigenstates of $H$ associated to a complex eigenvalue with negative imaginary part.
However, in general, $\H_\ads(H)$ may contain further states (unless $V$ is compactly supported, for instance; see \cite{FaFr18,FN}). 
We refer again to the next section for a more detailed discussion on the spectrum and spectral subspaces for $H$ and $H^*$.

Back to the self-adjoint setting, it is known that the number of bound states of the perturbed system
can be evaluated on scattering data. This relation bears the name of \emph{Levinson's theorem}
and was discovered by Levinson in \cite{Lev} in the context of
a Schr\"odinger operator with a spherically symmetric potential.
Similar investigations have then been generalized in numerous papers, and
we only mention the two review papers \cite{Ma,Ri} which contain several
earlier references. On a very specific example, the self-adjoint setting
was extended in \cite{NPR} where it is shown that complex eigenvalues can also
be counted thanks to Levinson's theorem. For more general dissipative systems, 
it was not known if an analog of Levinson's theorem holds, and if it does
what exactly is counted~?

As shown in the following sections, not only Levinson's theorem holds for dissipative systems in $\R^3$,
but a precise formula can be exhibited, namely
\begin{equation*}
\frac{1}{2\pi i} \int_0^\infty \d\lambda \;\! \left( \Tr\big(S(\lambda)^{-1}S'(\lambda)\big)+ c \lambda^{-\frac12} \right)   
= - \dim\big(\H_\ads(H)\big), 
\end{equation*}
with $c=\frac{i}{4 \pi} \int_{\R^3} \d x\;\! V(x)$, and the integral on the l.h.s.~understood as an improper Riemann integral. 
The precise statement is given in Theorem \ref{thm:analytic-formula}.
It means that the dimension of the 
vector space generated by the asymptotically disappearing states can be computed based on the scattering
matrix only. Note that, in this specific setting, no correction due to threshold effects appears (also called
half-bound states). On the other hand, the factor $c\lambda^{-\frac12}$ is quite common for scattering
in $\R^3$, even if $V$ is real-valued.

An interesting consequence of the previous equality, which appears to be new for non-compactly supported potentials, 
is that under our assumptions we have $\H_\ads(H)=\H_\disc(H)$, 
namely the asymptotically disappearing states coincide with the linear combinations of generalized eigenstates of $H$. 
This follows from the fact $\H_\ads(H)$ is finite-dimensional, by the previous equality. 

The proof of the above equality is provided in Section \ref{sec_topol}. It relies both on a $C^*$-algebraic
framework and on some deep analysis for the computation of a generalized winding number.
The key ingredient is an explicit formula for the wave operator $W_-$ which is derived in Section \ref{sec_W}.
Finally, general information on scattering theory in the dissipative setting and for Schr\"odinger operators
is provided in Section \ref{sec_frame}. The reader familiar with some of these topics can easily skip the corresponding section.
However, we provide all details, since dissipative systems have been less studied than their self-adjoint counterpart.

\section{Framework and scattering theory}\label{sec_frame}
\setcounter{equation}{0}

In the Hilbert space $\H:=\ltwo(\R^3)$ we consider the dissipative operator $H$ defined by
\begin{equation*}
H:= H_0 + V
\end{equation*}
where $H_0$ corresponds to the Laplace operator $-\Delta$ and where $V$ is the multiplication operator by a measurable and bounded function  $V:\R^3\to \C$ satisfying $\Im V \leq 0$. Here $\Im V\leq 0$ means that the imaginary part of $V(x)$ is non-positive for almost every $x\in \R^3$.
Additional decay conditions on $V$ will be imposed on due time.
Clearly, since $V$ is bounded, the domain of $H$ corresponds to the domain of $H_0$, namely the 
second Sobolev space on $\R^3$. 

We recall that an operator $A$ is called dissipative if $\Im (\langle f , A f \rangle)\le 0$ for all $f\in\dom(A)$, and maximal dissipative if it is dissipative and has no proper dissipative extension. We denote by $\C_-:=\{z\in\C\, | \, \Im z\le 0\}$ the complex lower-half plane.

We begin with the following two easy lemmas whose proofs are recalled for completeness. 
 
\begin{Lemma}\label{lm:first}
Suppose that $V:\R^3\to \C_-$ is a bounded, measurable function. The following properties hold.
\begin{enumerate}[label=(\roman*)]
\item $H$ is a maximal dissipative operator on $\H$ with domain
$
\dom( H ) = \dom( H_0 ) .
$
\item The operator $-i H$ generates a strongly continuous group $\{ \e^{-itH} \}_{ t \in \R}$ which satisfies
\begin{equation*}
\big \|\e^{-itH}\big\| \le 1 \, \hbox{ if } \,  t \ge 0 , \qquad \big\|\e^{-itH}\big\| \le \e^{\|V\|_{L^\infty}|t|} \, \hbox{ if } \, t \le 0 .
\end{equation*}
In particular, $-i H$ generates the strongly continuous semigroup of contractions $\{ \e^{-itH} \}_{ t \ge 0 }$.
\item The adjoint of $H$ is
\begin{equation*}
H^* = H_0 + \overline{V},
\end{equation*}
with domain $\dom( H^* ) = \dom( H_0 )$, where $\overline{V}$ is the complex conjugate of $V$. Moreover, $i H^*$ generates of a strongly continuous group $\{ \e^{itH^*} \}_{t \in \R}$ such that $\{ \e^{itH^*} \}_{t \ge 0}$ is a semigroup of contractions.
\end{enumerate}
\end{Lemma}

\begin{proof}
$(i)$ Since $V$ is bounded, the operator $H$ is closed with domain $\dom( H ) = \dom( H_0 )$. Moreover, $H$ is dissipative since, for all $f \in \dom( H )$,
\begin{equation*}
\Im( \langle f , H f \rangle ) = \Im ( \langle f , V f \rangle ) \le 0,
\end{equation*}
by assumption. To verify that $H$ is maximal dissipative, it then suffices to show that $\Ran( H - i \lambda ) = \H$ for some $\lambda > 0$, see \cite[Corol.~p.~201]{Ph59_01}. This easily follows from the identity
\begin{equation*}
H - i \lambda = H_0 - i \lambda + V =  (H_0-i\lambda) \big(1+( H_0 - i \lambda )^{-1} V \big ).
\end{equation*}
Indeed, since $H_0$ is self-adjoint, $(H_0-i\lambda)$ is invertible, and we have the estimate $\|( H_0 - i \lambda)^{-1} V\|\le\lambda^{-1}\|V\|_{L^\infty}$ which implies that $1+( H_0 - i \lambda )^{-1} V$ is also invertible for $\lambda>\|V\|_{L^\infty}$.

$(ii)$ Since $H_0$ is self-adjoint, $-iH_0$ generates the strongly continuous unitary group $\{ \e^{-itH_0} \}_{ t \in \R}$. Hence, since $V$ is bounded, a perturbation argument (see, e.g., \cite[Thm.~11.4.1]{Da07_01}) shows that $-iH$ generates a strongly continuous group $\{ \e^{-itH} \}_{ t \in \R}$ satisfying 
$$
\big\|\e^{-itH} \big\| \le \e^{\|V\|_{L^\infty}|t|},
$$
 for all $t \in \R$. The fact that $\e^{-itH}$ is a contraction for $t \ge 0$ is a consequence of the fact that $H$ is maximal dissipative (see e.g. \cite[Thm.~10.4.2]{Da07_01}).

$(iii)$ It easily follows from the definition of $H^*$ and direct computations that $H^* = H_0 + \overline{V}$ with domain $\dom( H^* ) = \dom( H_0 )$. The same argument as in $(ii)$ then yields $(iii)$.
\end{proof}

The next lemma concerns the spectrum of $H$. Since $H$ is dissipative, one easily verifies that its spectrum, $\sigma(H)$, is contained in the lower half-plane $\{ z\in\C \mid \Im z\le0\}$. We now assume that $V$ vanishes at $\infty$ and that $\Im V<0$ on a non-trivial open set. The essential spectrum of $H$ is defined by
$$
\sigma_\ess(H):=\{\lambda\in\C \mid H-\lambda\text{ is not Fredholm}\},
$$
while the discrete spectrum is
$$
\sigma_\disc(H):=\{\lambda\in\C \mid \lambda \text{ is an isolated eigenvalue of } H \text{ and } \dim\,\Ran(\Pi_\lambda(H)) < \infty \} ,
$$
where $\Pi_\lambda(H)=(2i\pi)^{-1}\int_\gamma(z-H)^{-1}\d z$ stands for the usual Riesz projection associated to $\lambda$, with $\gamma$ a circle oriented counterclockwise and centered at $\lambda$, with sufficiently small radius.

\begin{Lemma}
Suppose that $V:\R^3\to \C_-$ is a bounded, measurable function such that $V(x)\to0$ as $|x|\to\infty$ and $\Im V<0$ on a non-trivial open set. The following properties hold:
\begin{enumerate}[label=(\roman*)]
\item $\sigma_\ess(H)=\sigma(H_0)=[0,\infty)$.
\item $\sigma_\disc(H)=\sigma(H)\setminus\sigma_\ess(H)$. Moreover $\sigma_\disc(H)$ consists of isolated eigenvalues with non-positive imaginary parts, of finite algebraic multiplicities, and that may only accumulate at $\sigma_\ess(H)$.
\item $H$ does not have real eigenvalues.
\end{enumerate}
\end{Lemma}
 
 \begin{proof}
Since $V$ is bounded and $V(x)\to0$ as $|x|\to\infty$, $V$ is a relatively compact perturbation of the self-adjoint operator $H_0$. This implies that $\sigma_\ess(H)=\sigma_\ess(H_0)$ and hence $(i)$ since $H_0=-\Delta$. 

For $(ii)$, by using \cite[Cor. 2 of Thm. XIII.14]{RS4} and \cite[Thms IX.1.6 and
IX.2.1]{EdEv87}, it follows that the different definitions of essential spectra for non self-adjoint operators considered in \cite[Chap.~IX]{EdEv87} all coincide. In particular we have $\sigma_\disc(H)=\sigma(H)\setminus\sigma_\ess(H)$ and $\sigma_\disc(H)$ is at most countable and can only accumulate at points of $\sigma_\ess(H)$.
It remains to show that any eigenvalue in $\sigma_\disc(H)$ has a finite algebraic multiplicity. Let $\lambda\in\sigma_\disc(H)$. Since the range of the Riesz projection $\Pi_\lambda(H)$ is finite, \cite[Thm. XII.5 (d)]{RS4} yields that $\Ran(\Pi_\lambda(H))$ is spanned by the generalized eigenvectors of $H$
associated to $\lambda$, i.e., by the vectors $\varphi\in\dom(H^k)$ such that $(H-\lambda)^k\varphi = 0$ for some integer $k$, which in turn implies that $\dim\,\Ran(\Pi_\lambda(H))$ coincides with the algebraic multiplicity of $\lambda$.

For $(iii)$, suppose that $\lambda\in\R$ is an eigenvalue of $H$ and let $\varphi\in\dom(H)$, $\varphi\neq0$, be such that $(H-\lambda)\varphi=0$. Taking the scalar product with $\varphi$ and then the imaginary part, we obtain
$$
\int_{\R^3}\d x\, (\Im V)(x)|\varphi(x)|^2=0,
$$
and hence, since $\Im V\le0$, $\varphi=0$ on $\supp(\Im V)$. Since $\supp(\Im V)$ contains a non-trivial open set, by the unique continuation principle (see e.g. \cite[Thm.~XIII.63]{RS4}), this yields $\varphi=0$, which is a contradiction.
 \end{proof}
 
 Several spectral subspaces can be naturally introduced for dissipative operators in Hilbert spaces. We refer to \cite{FaFr23} for a thorough discussion in the more general context of non-self-adjoint operators. Here we only mention the discrete spectral subspace
 \begin{equation*}
 \H_\disc(H):=\Span \big\{ \varphi \in \Ran(\Pi_\lambda(H))\mid \lambda\in\sigma_\disc(H) \big \}^\cl,
 \end{equation*}
where for a set $\Omega$, $\Omega^\cl$ stands for the closure of $\Omega$. Recalling that the space of asymptotically disappearing states $\H_\ads(H)$ has been introduced in \eqref{eq:def_Hads}, we have the following easy lemma which follows directly from \cite[Prop. 5.6]{FaFr23}.

\begin{Lemma}\label{lm:Hdisc-sub-Hads}
Suppose that $V:\R^3\to \C_-$ is a bounded, measurable function such that  $V(x)\to0$ as $|x|\to\infty$ and $\Im V<0$ on a non-trivial open set. Then
$$
\H_\disc(H)\subset\H_\ads(H).
$$
\end{Lemma}

It is also proved in \cite[Thm. 6.1]{FaFr18} that the equality $\H_\disc(H)=\H_\ads(H)$ holds when $V$ is compactly supported 
and $0$ is neither an eigenvalue nor a resonance of $-\Delta+\Re V$. 
Our main result provides different conditions ensuring that $\H_\disc(H)=\H_\ads(H)$, using the following easy lemma.

\begin{Lemma}\label{lem:finite-dim}
Under the conditions of Lemma \ref{lm:Hdisc-sub-Hads} and if $\H_\ads(H)$ is finite dimensional, then
$$
\H_\disc(H)=\H_\ads(H).
$$
\end{Lemma}
\begin{proof}
Clearly, $\H_\ads(H)$ is an invariant subspace for $H$. Denote by $\tilde H$ the restriction of $H$ to $\H_\ads(H)$. \
For all $\varphi\in\H_\ads(H)$, we have $\|\e^{-it \tilde H}\varphi\|=\|\e^{-itH}\varphi\|\to0$ as $t\to\infty$. 
Since $\H_\ads(H)$ is finite dimensional, it follows from Liapunov's Theorem, see e.g. \cite[Thm. 2.10]{EngelNagel}, that 
all eigenvalues of $\tilde H$ have a negative imaginary part. 
Therefore, any $\varphi$ in $\H_\ads(H)$ is a finite linear combination of generalized eigenstates of $\tilde H$, and hence of $H$, 
associated to eigenvalues with negative imaginary parts. 
In other words, $\H_\ads(H)\subset \H_\disc(H)$. Together with Lemma \ref{lm:Hdisc-sub-Hads}, this gives the desired equality.
\end{proof}

Our aim is to study the wave operator and the scattering operators for the pair $(H,H_0)$.
Recall that for dissipative operators they are defined by the time-dependent expressions \eqref{eq:waveintro} and \eqref{eq:scatt_operator_intro}.
By Lemma \ref{lm:first}, when they exist, these operators are contractions.
For subsequent investigations it will be necessary to have stationary expressions for these operators.
For that purpose and for $z\in \C\setminus \R$ let us denote by 
$R_0(z)$ the resolvent $(H_0-z)^{-1}$ of $H_0$. For  $\lambda \in \R$ and $\varepsilon >0$ we also set $R(\lambda+i\varepsilon)$ for $(H-\lambda -i\varepsilon)^{-1}$
and $R^*(\lambda-i\varepsilon)$ for $(H^*-\lambda +i\varepsilon)^{-1}$.
The proof of the next lemma is fairly standard for self-adjoint operators. It is not difficult to adapt it to the setting studied in this paper. For the convenience of the reader, we provide the details.

\begin{Lemma}\label{lem_statio}
If the wave operators $W_\pm$ exist, then for any $f, g\in \H$ the following equalities hold:
$$
\langle W_+ f,g\rangle 
=  \lim_{\varepsilon\searrow 0}\frac{\varepsilon}{\pi} \int_{\R}\d \lambda 
\langle R_0(\lambda+i \varepsilon)f, R(\lambda + i \varepsilon) g\rangle
$$
and 
$$
\langle W_- f,g\rangle
=  \lim_{\varepsilon\searrow 0}\frac{\varepsilon}{\pi} \int_{\R}\d \lambda 
\langle R_0(\lambda- i \varepsilon)f, R^*(\lambda - i \varepsilon) g\rangle.
$$
\end{Lemma}

\begin{proof}
We give the proof for $W_+$, the proof for $W_-$ is similar. 
Since $\{ \e^{-itH} \}_{ t \ge 0 }$ is a contraction semigroup and $\{ \e^{-itH_0} \}_{ t \in\R }$ is a unitary group, we have 
\begin{equation*}
R(\lambda+i\varepsilon)f=i\int_0^\infty \d t\, \e^{-\varepsilon t}\e^{-it(H-\lambda)}f, \qquad R_0(\lambda+i\varepsilon)f=i\int_0^\infty \d t\, \e^{-\varepsilon t}\e^{-it(H_0-\lambda)}f,
\end{equation*}
for all $\varepsilon>0$ and $f\in\H$ with $\|f\|=1$, and 
\begin{equation*}
\frac{\varepsilon}{\pi} \int_{\R}\d \lambda \|R(\lambda+i\varepsilon)f\|^2\le1, \qquad \frac{\varepsilon}{\pi} \int_{\R}\d \lambda \|R_0(\lambda+i\varepsilon)f\|^2=1.
\end{equation*}
Let $\varepsilon>0$ and $f,g\in\H$. Using Lebesgue's dominated convergence Theorem, we write
\begin{align}\label{eq:int-formula}
&\frac{\varepsilon}{\pi} \int_{\R}\d \lambda \langle R_0(\lambda+i \varepsilon)f, R(\lambda + i \varepsilon) g\rangle =\frac{\varepsilon}{\pi}\lim_{\delta\searrow 0}  \int_{\R}\d \lambda \e^{-\delta\lambda^2} \langle R_0(\lambda+i \varepsilon)f, R(\lambda + i \varepsilon) g\rangle , 
\end{align}
and then compute, for all $\delta>0$,
\begin{align*}
&\int_{\R}\d \lambda \e^{-\delta\lambda^2} \langle R_0(\lambda+i \varepsilon)f, R(\lambda + i \varepsilon) g\rangle \\
&= \int_{\R}\d \lambda \int_0^\infty \d t \int_0^\infty \d t' \e^{-\delta\lambda^2} \e^{-\varepsilon (t+t')} \e^{i\lambda(t'-t)} \langle \e^{-itH_0}f, \e^{-it'H} g\rangle \\
&= \int_{\R}\d \lambda \int_0^\infty \d t \int_{-t}^\infty \d s \,  \e^{-\delta\lambda^2} \e^{-2\varepsilon t} \e^{-\varepsilon s} \e^{i\lambda s} \langle \e^{-itH_0}f, \e^{-i(t+s)H} g\rangle ,
\end{align*}
where we used the change of variables $t'=t+s$ in the last equality. Using Fubini's Theorem and computing the integral in $\lambda$, we obtain
\begin{align*}
&\int_{\R}\d \lambda \e^{-\delta\lambda^2} \langle R_0(\lambda+i \varepsilon)f, R(\lambda + i \varepsilon) g\rangle \\
&= (2\pi)^\frac12 (2\delta)^{-\frac12} \int_0^\infty \d t \int_{-t}^\infty \d s \,  \e^{-s^2/(4\delta)} \e^{-2\varepsilon t} \e^{-\varepsilon s}  \langle \e^{-itH_0}f, \e^{-i(t+s)H} g\rangle \\
&= (2\pi)^{\frac12} 2^\frac12 \int_0^\infty \d t \int_{-(4\delta)^{-1/2}t}^\infty \d \tau \,  \e^{-\tau^2} \e^{-2\varepsilon t} \e^{-(4\delta)^{1/2}\varepsilon \tau}  \langle \e^{-itH_0}f, \e^{-itH} \e^{-i(4\delta)^{1/2}\tau H} g\rangle ,
\end{align*}
where we used the change of variables $\tau=(4\delta)^{-1/2}s$. Applying again Lebesgue's dominated convergence Theorem and then computing the integral in $\tau$, we obtain
\begin{align*}
\lim_{\delta\searrow 0}\int_{\R}\d \lambda \e^{-\delta\lambda^2} \langle R_0(\lambda+i \varepsilon)f, R(\lambda + i \varepsilon) g\rangle &= 2\pi \int_0^\infty \d t \e^{-2\varepsilon t}  \langle \e^{-itH_0}f, \e^{-itH} g\rangle.
\end{align*}
Inserting this into \eqref{eq:int-formula} yields
\begin{align*}
\frac{\varepsilon}{\pi} \int_{\R}\d \lambda \langle R_0(\lambda+i \varepsilon)f, R(\lambda + i \varepsilon) g\rangle &= 2\varepsilon \int_0^\infty \d t \e^{-2\varepsilon t}  \langle \e^{-itH_0}f, \e^{-itH} g\rangle\\
&= 2 \int_0^\infty \d s \e^{-2s}  \langle \e^{-i\varepsilon^{-1}sH_0}f, \e^{-i\varepsilon^{-1}sH} g\rangle \\
&= 2 \int_0^\infty \d s \e^{-2s}  \langle \e^{i\varepsilon^{-1}sH^*} \e^{-i\varepsilon^{-1}sH_0}f,  g\rangle ,
\end{align*}
by the change of variables $s=\varepsilon t$. A last application of Lebesgue's dominated convergence Theorem gives the result.
\end{proof}

Let us now use a convenient decomposition of the potential $V$, namely
$$
V=u v^2
$$
with $v(x):=|V(x)|^{1/2}$ for all $x\in \R^3$, and $u(x):=\frac{V(x)}{v(x)^2}$ if $v(x)\neq 0$
while $u(x)=1$ if $v(x)=0$. 
If we set 
$$
\C_-:=\{z\in \C\mid \Im(z)\leq 0\}
$$
then our condition on $V$ imposes that $u(x)\in \C_-$ for any $x\in \R^3$.

In the sequel, the following equalities will be constantly used, namely for any $z\in  \C \setminus \sigma(H)$
\begin{align}\label{eq_resolv_eq}
\nonumber R(z)
& = R_0(z)-R_0(z) v\big(u- uvR(z)vu\big)vR_0(z) \\
& = R_0(z)-R_0(z) v \big(\overline{u} +v R_0(z)v\big)^{-1}vR_0(z), 
\end{align}
where $\overline{u}$ denotes the complex conjugation of the function $u$.
These equalities can be easily deduced from the resolvent equations.
For the following statement we need to be more precise about the decay of the potential
at infinity. For that purpose, we assume that
\begin{equation}\label{eq_decay_V}
|V(x)|\leq c \langle x\rangle^{-\alpha}
\end{equation}
for some $c>0$, $\alpha>0$ and a.e.~$x\in \R^3$.

As usual, $\B(\H)$ stands for the set of bounded operators on the Hilbert space $\H$.

\begin{Lemma}\label{lm:properties-R}
 Suppose that $V:\R^3\to \C_-$ is a bounded, measurable function satisfying \eqref{eq_decay_V} with $\alpha>1$ and $\Im V<0$ on a non-trivial open set. For any $\lambda >0$ the limits
$$
v R_0(\lambda \pm i0)v := \lim_{\varepsilon \searrow 0} v R_0(\lambda \pm i\varepsilon)v\quad\hbox{ and }\quad v R(\lambda + i0)v := \lim_{\varepsilon \searrow 0} v R(\lambda + i\varepsilon)v
$$
exist in the norm topology of $\B(\H)$ and the maps
$$
(0,\infty)\ni \lambda \mapsto v R_0(\lambda\pm i0)v \in \B(\H) \quad \hbox{ and } \quad (0,\infty)\ni \lambda \mapsto v R(\lambda+i0)v \in \B(\H)
$$
are continuous. Moreover $\overline{u} +v R_0(\lambda+i0)v$ and $u- uvR(\lambda+i0)vu$ are invertible in $\B(\H)$ and the maps
$$
(0,\infty)\ni \lambda \mapsto \big(\overline{u} +v R_0(\lambda+i0)v\big)^{-1} \in \B(\H) \quad\hbox{ and }\quad (0,\infty)\ni \lambda \mapsto \big(u- uvR(\lambda+i0)vu\big)^{-1} \in \B(\H)
$$
are continuous.

If $\alpha>2$ in \eqref{eq_decay_V}, then 
$$
v R_0(+i0)v := \lim_{\varepsilon \searrow 0} v R_0(i\varepsilon)v\quad\hbox{ and }\quad v R(+i0)v := \lim_{\varepsilon \searrow 0} v R(i\varepsilon)v
$$
also exist in the norm topology of $\B(\H)$, $\overline{u} +v R_0(+i0)v$ and $u- uvR(+i0)vu$ are invertible in $\B(\H)$ and the maps
\begin{align}
&[0,\infty)\ni \lambda \mapsto v R_0(\lambda+i0)v \in \B(\H) \quad \hbox{ and } \quad [0,\infty)\ni \lambda \mapsto v R(\lambda+i0)v \in \B(\H), \label{eq:bounded1} \\
&[0,\infty)\ni \lambda \mapsto \big(\overline{u} +v R_0(\lambda+i0)v\big)^{-1}\in \B(\H)\quad\hbox{ and }\quad [0,\infty)\ni \lambda \mapsto \big(u- uvR(\lambda+i0)vu\big)^{-1} \in \B(\H) \label{eq:bounded2}
\end{align}
are continuous and bounded.
\end{Lemma}

In the next proof and in the sequel, we use the standard notation $\langle X\rangle$ for the multiplication
operator by the function $x\mapsto (1+x^2)^{1/2}$.

\begin{proof}
Since $v(x)\to0$ as $|x|\to\infty$ by \eqref{eq_decay_V}, the operators $vR_0(\lambda+i\varepsilon)v$ are compact for all $\varepsilon>0$. Moreover it follows from \eqref{eq_decay_V} that 
\begin{equation*}
v(x)\leq c \langle x\rangle^{-\frac{\alpha}{2}}.
\end{equation*}
It is then well-known that the limit 
$$
v R_0(\lambda + i0)v := \lim_{\varepsilon \searrow 0} v R_0(\lambda + i\varepsilon)v
$$
exists in $\B(\H)$. It is in addition a compact operator, as a norm-limit of compact operators. 

In order to prove that $\overline{u} +v R_0(\lambda+i0)v$ is invertible in $\B(\H)$, it suffices to show that $1 +v R_0(\lambda+i0)vu$ is invertible in $\B(\H)$, i.e. that $\lambda$ is not an outgoing spectral singularity, in the terminology of \cite{FaFr23}.  All statements of the lemma then follow from \cite[Prop.~3.7]{FaFr23}, except for the boundedness of the maps in \eqref{eq:bounded1}--\eqref{eq:bounded2}. To verify this it suffices to verify boundedness in a neighborhood of $\infty$. This easily follows from the resolvent equation and the well-known fact that
$$
\|\langle X\rangle^{-s}R_0(\lambda+i0)\langle X\rangle^{-s}\|= \O(\lambda^{-\frac12}),\quad\lambda\to\infty,
$$
for any $s>\frac12$.
\end{proof}

The following properties of the wave operators have been derived in \cite[Prop. 3.4 and 3.5]{FaFr18} in an abstract setting, using assumptions that are not necessarily satisfied in our context. It is however straightforward to adapt the proof, we only emphasize the differences here.

\begin{Lemma}\label{lem:existence_wave}
Suppose that $V:\R^3\to \C_-$ is a bounded, measurable function satisfying \eqref{eq_decay_V} with $\alpha>2$ and $\Im V<0$ on a non-trivial open set.
Then the wave operators $W_\pm$ exist, are injective contractions, and satisfy
\begin{equation*}
\Ran(W_+)^{\cl}=\H_\ads(H)^\perp, \qquad \Ran(W_-)^{\cl}=\H_\ads(H^*)^\perp.
\end{equation*}
\end{Lemma}

\begin{proof}
We prove the lemma for $W_+$, the proof for $W_-$ is identical.
Writing
\begin{equation*}
\e^{itH^*}\e^{-itH_0}\varphi=\varphi+i\int_0^t\d\tau\, \e^{i\tau H^*}\overline{u}v^2\e^{-i\tau H_0}\varphi,
\end{equation*}
for any $\varphi\in \H$, one sees that it suffices to show that the integral converges as $t\to\infty$. For $0<t_1<t_2<\infty$, we have
\begin{align}
\Big\|\int_{t_1}^{t_2}\d\tau\, \e^{i\tau H^*}\overline{u}v^2\e^{-i\tau H_0}\varphi\Big\|&
\le\sup_{\psi\in\H,\|\psi\|=1} \int_{t_1}^{t_2}\d\tau\big|\big\langle v\e^{-i\tau H}\psi,\overline{u}v\e^{i\tau H_0}\varphi\big\rangle\big|\notag\\
&\le\sup_{\psi\in\H,\|\psi\|=1}\Big(\int_{t_1}^{t_2}\d\tau\big\|v\e^{-i\tau H}\psi\big\|^2\Big)^{\frac12}\Big(\int_{t_1}^{t_2}\d\tau
\big\|v\e^{i\tau H_0}\varphi\big\|^2\Big)^{\frac12},\label{eq:bound_int}
\end{align}
where we used the Cauchy-Schwarz inequality in the second inequality. Since $\alpha >1$ in \eqref{eq_decay_V}, it is well-known, 
using the explicit expression of the kernel of $R_0(\lambda+i0)$, that there exists $C_v>0$ such that
\begin{align}\label{eq:kato-bound-H0}
\int_{0}^\infty \d \tau \|v\e^{\pm i\tau H_0}\varphi\|^2=\frac{1}{\pi}\int_{\R} \d \lambda \|vR_0(\lambda\mp i0)\varphi\|^2 \le C_v\|\varphi\|^2.
\end{align}
Moreover, using the resolvent equation, we also have
\begin{align*}
\int_{0}^\infty \d \tau \|v \e^{-i\tau H}\psi\|^2&=\frac{1}{\pi}\int_{\R} \d \lambda \|vR(\lambda+i0)\psi\|^2 \\
&\le \frac{2}{\pi}\int_{\R} \d \lambda \|vR_0(\lambda+i0)\psi\|^2 + \frac{2}{\pi}\int_{\R} \d \lambda \|vR(\lambda+i0)v u vR_0(\lambda+i0) \psi\|^2 .
\end{align*}
Since $[0,\infty)\ni\lambda\mapsto vR(\lambda+i0)v\in\B(\H)$ is bounded by Lemma \ref{lm:properties-R},  
the estimate \eqref{eq:kato-bound-H0} implies that
\begin{align}\label{eq:kato-bound-H}
\int_{0}^\infty \d \tau \|v \e^{-i\tau H}\psi\|^2&\le C'_v\|\psi\|^2,
\end{align}
for some $C'_v>0$. 
From \eqref{eq:bound_int}, \eqref{eq:kato-bound-H0} and \eqref{eq:kato-bound-H}, 
we deduce the existence of $W_+$.

Since $\e^{itH^*}$ is a contraction and $\e^{-itH_0}$ is unitary, $W_+$ is clearly a contraction. 
To prove the injectivity, one can proceed as in \cite[Prop. 3.4]{FaFr18}. 
Likewise, the equality $\Ran(W_+)^{\cl}=\H_\ads(H)^\perp$ follows as in \cite[Prop. 3.5]{FaFr18}. 
\end{proof}

For the study of the scattering operator, defined by \eqref{eq:scatt_operator_intro}, let us introduce the diagonalization of $H_0$.
For this, let $\SS$ be the Schwartz space on $\R^3$, let $\R_+:=(0,\infty)$, 
and let
$$
\HS:=\ltwo(\S^2).
$$
We denote by $\F$ the usual Fourier transform
in $\R^3$ (which is unitary on $\ltwo(\R^3)$), and set for $f\in \SS$, $\lambda>0$ and $\omega\in \S^2$
$$
\big[[\F_0 f](\lambda)\big](\omega):= \big(\tfrac{\lambda}{4}\big)^{1/4}[\F f](\sqrt{\lambda}\omega).
$$
It is known that $\F_0:\H\to \ltwo(\R_+, \d \lambda; \HS)$ extends to a unitary map
satisfying 
$$
\big[[\F_0 H_0 \F_0^*\varphi](\lambda)\big]( \omega) = \lambda [\varphi(\lambda)](\omega)
$$
for suitable $\varphi \in \ltwo(\R_+, \d \lambda; \HS)$.
Let us also recall from \cite[Sec.~3]{Jen81} that the operator
$\F_0(\lambda):\SS\to\HS$ given by $\F_0(\lambda)f:=(\F_0f)(\lambda)$ extends to an
element of $\B(\H^s_t,\HS)$ for each $s\in\R$ and $t>1/2$, and the map
$\R_+\ni\lambda\mapsto\F_0(\lambda)\in\B(\H^s_t,\HS)$ is continuous.
Here we have used $\H^s_t$ for the weighted Sobolev spaces
over $\R^3$ with index $s\in\R$ associated to derivatives and index $t\in\R$
associated to decay at infinity, with the convention that
$\H^s:=\H^s_0$ and $\H_t:=\H^0_t$.

The next statement can be directly inferred from \cite[Lem.~2.1 \& Lem.~2.2]{RT13}. 
We denote by $\B(\H,\HS)$ and $\K(\H,\HS)$ the set of bounded and compact operators from $\H$ to $\HS$, respectively.

\begin{Lemma}\label{lem_Fpm}
Let  $s>3/2$.  Then, the functions
$$
(0,\infty)\ni\lambda\mapsto\lambda^{\pm1/4}\F_0(\lambda) \langle X\rangle^{-s} \in\B(\H,\HS)
$$
are continuous and bounded. In addition,  the map 
$$
\R_+\ni \lambda \mapsto \lambda^{-1/4} \F_0(\lambda) \langle X\rangle^{-s} \in \K(\H,\HS)
$$
is continuous, admits a limit as $\lambda\searrow0$ and vanishes as $\lambda\to\infty$.
\end{Lemma}

With the notation introduced above, let us prove a stationary representation of the
scattering matrix, which extends the result provided in \cite[Eq.~(2.18)]{FN}.
Indeed, it is known, by usual intertwining properties, that the scattering operator $S$ defined in \eqref{eq:scatt_operator_intro} commute with $H_0$, and therefore
 satisfies for a.e.~$\lambda>0$
$$
[\F_0 S \F_0^*\varphi](\lambda,\omega) =\big[S(\lambda)\varphi(\lambda)\big](\omega),
$$
where $S(\lambda)$ is an operator acting on $\HS$.
In the sequel, we denote by $\L^1(\HS)$ the set of trace-class operator on the Hilbert space $\HS$.

\begin{Lemma}\label{lm:S(lambda)}
Suppose that $V:\R^3\to \C_-$ is a bounded, measurable function satisfying \eqref{eq_decay_V} with $\alpha>2$ and $\Im V<0$ on a non-trivial open set. Then the following equality holds for all $\lambda>0$:
\begin{equation}\label{eq_S_stationary}
S(\lambda)= 1-2\pi i \F_0(\lambda) v \big(\overline{u} +v R_0(\lambda+i0)v\big)^{-1}v \F_0(\lambda)^*.
\end{equation}
Moreover, for all $\lambda>0$, $S(\lambda)$ is a contraction such that $\|S(\lambda)\|_{\B(\HS)}=1$, $S(\lambda)-1$ is compact, and the map $(0,\infty)\ni\lambda\mapsto S(\lambda)\in\B(\HS)$ is continuous.

If the decay condition is strengthened to $\alpha >3$ in \eqref{eq_decay_V}, then $[0,\infty)\ni\lambda\mapsto S(\lambda)\in\B(\HS)$ is continuous and we have
\begin{equation*}
S(0)=1,\qquad \lim_{\lambda\to\infty}\|S(\lambda)-1\|_{\B(\HS)}=0.
\end{equation*}

If $\alpha>4$ in \eqref{eq_decay_V}, then 
$S(\lambda)-1 \in \L^1(\HS)$ for all $\lambda \in [0,\infty)$ and the map $\lambda \mapsto S(\lambda)-1$ is continuously differentiable in $\L^1(\HS)$.
\end{Lemma}

\begin{proof}
Assuming that $\alpha >2$ in \eqref{eq_decay_V}, one can mimic the proof of \cite[Thm.~2.6]{FN}, using \eqref{eq:kato-bound-H0} and \eqref{eq:kato-bound-H} instead of Eqs.~(4.1) and (4.2) in \cite{FN}. This leads to the statement of the lemma concerning the map $(0,\infty)\ni \lambda \mapsto S(\lambda)$.

If $\alpha>3$ in \eqref{eq_decay_V}, then combining Lemma \ref{lm:properties-R} and Lemma \ref{lem_Fpm}, we deduce that $\lambda\mapsto S(\lambda)$ 
is also continuous at $\lambda=0$ and satisfies $S(0)=1$. 
Combining moreover Lemmas \ref{lm:properties-R} and \ref{lem_Fpm}, one easily verifies that $\lim_{\lambda\to\infty}\|S(\lambda)-1\|_{\B(\HS)}=0$. 

Assume now $\alpha>4$ in \eqref{eq_decay_V}.
We refer to \cite[Prop.~8.1.9]{Yaf10} for a proof, in the self-adjoint
case, that $S(\lambda)-1\in \L^1(\HS)$ and $[0,\infty)\ni\lambda\mapsto S(\lambda)-1$ is continuously differentiable. 
This can be straightforwardly adapted in our setting.
\end{proof}

An important difference between the usual unitary scattering theory and the dissipative scattering theory that we consider here is that the scattering matrix $S(\lambda)$ is only a contraction (not a unitary operator) and hence it may not be invertible on $\HS$. An energy $\lambda>0$ such that $S(\lambda)$ is not invertible is called a spectral singularity. The reader is referred to \cite{FaFr18,FN,FaFr23}, where the notion of spectral singularity plays a central role, for various characterizations of this notion in an abstract context. Note that the previous lemma shows that the threshold energy $0$ cannot be a spectral singularity if $V$ decays fast enough; see also  \cite[Thm.~1.1]{Wa11} or \cite[Prop.~3.7]{FaFr23}. 
A characterization of spectral singularity in terms of real resonance is also possible, see
for example  \cite[Thm.~3.6]{FaFr23}. 

In the next statement we provide sufficient conditions for the invertibility of the scattering
matrix.

\begin{Lemma}\label{lem_S_invert}
Suppose that $V:\R^3\to \C_-$ is a bounded, measurable function satisfying \eqref{eq_decay_V} with $\alpha>3$ and $\Im V<0$ on a non-trivial open set. 
Let $\lambda > 0$ be such that the operator $\overline{u} + vR_0(\lambda-i0)v$
is invertible in $\B(\H)$. Then the operator $S(\lambda)$ is invertible in $\B(\HS)$, 
and the map
\begin{equation*}
[0,\infty)\ni\lambda\mapsto S(\lambda)^{-1}\in\B(\HS)
\end{equation*}
is continuous. Moreover, if $\overline{u} + vR_0(\lambda-i0)v$
is invertible in $\B(\H)$ for all $\lambda>0$, then $S$ is invertible in $\B(\H)$ and $W_\pm$ have closed ranges.
\end{Lemma}

\begin{proof}
We firstly recall the relation $2\pi i \F_0(\lambda)^*\F_0(\lambda) = R_0(\lambda+i0)-R_0(\lambda-i0)$ and make the computation
\begin{align*}
&S(\lambda)\big(1+2\pi i \F_0(\lambda)v\big(\overline{u} + vR_0(\lambda-i0)v\big)^{-1} v \F_0(\lambda)^*) \\
&= 1 -2\pi i \F_0(\lambda) v \big(\overline{u} +v R_0(\lambda+i0)v\big)^{-1}v \F_0(\lambda)^* +2\pi i \F_0(\lambda)v\big(\overline{u} + vR_0(\lambda-i0)v\big)^{-1} v \F_0(\lambda)^* \\
&-(2\pi i)^2 \F_0(\lambda) v\big(\overline{u}+vR_0(\lambda+i0)v\big)^{-1}v\F_0(\lambda)^*\F_0(\lambda)v\big(\overline{u} + v R_0(\lambda -i0)v\big)^{-1} v\F_0(\lambda)^* \\
&= 1 - 2\pi i \F_0(\lambda) v \big( \big(\overline{u}+vR_0(\lambda+i0)v\big)^{-1}-\big(\overline{u} + vR_0(\lambda-i0)v\big)^{-1} \big) v\F_0(\lambda)^* \\
&- 2\pi i \F_0(\lambda) v \big(\overline{u} +v R_0(\lambda+i0)v\big)^{-1}v \big(R_0(\lambda+i0)-R_0(\lambda-i0)  \big) v \big(\overline{u} + vR_0(\lambda-i0)v\big)^{-1} v \F_0(\lambda)^* \\
&= 1,
\end{align*}
so that $1+2\pi i \F_0(\lambda)v\big(\overline{u} + vR_0(\lambda-i0)v\big)^{-1} v \F_0(\lambda)^*$ is a right inverse for $S(\lambda)$. A similar calculation shows that it is a left inverse and so $S(\lambda)$ is invertible with inverse given by
\begin{align*}
S(\lambda)^{-1} &= 1+2\pi i \F_0(\lambda)v\big(\overline{u} + vR_0(\lambda-i0)v\big)^{-1} v \F_0(\lambda)^*.
\end{align*}

If $\bar u +vR_0(\lambda-i0)v$ is invertible for all $\lambda>0$, then it follows from Lemma \ref{lm:properties-R} that the map 
\begin{equation*}
[0,\infty)\ni\lambda\mapsto\big(\overline{u} + vR_0(\lambda-i0)v\big)^{-1}\in\B(\H)
\end{equation*}
is continuous. Together with Lemma \ref{lem_Fpm}, this proves the continuity of $\lambda\mapsto S(\lambda)^{-1}$ on $[0,\infty)$.
Arguing as in the proof of Lemma \ref{lm:S(lambda)}, we have $\|S(\lambda)^{-1}-1\|_{\B(\HS)}\to 0$ as $\lambda\to\infty$. 
Therefore $\lambda\mapsto S(\lambda)^{-1}$ is bounded on $[0,\infty)$, which implies that $S$ is invertible in $\B(\H)$.

Finally, to prove that $W_-$ has a closed range, it suffices to write, for any $\varphi\in\H$
\begin{align*}
\|W_-\varphi\|\ge\|W_+^*W_-\varphi\|=\|S\varphi\|\ge\|S^{-1}\|^{-1}\|\varphi\|,
\end{align*}
where we used that $W_+$ is a contraction. The proof for $W_+$ is similar.
\end{proof}

\section{About the wave operator $W_-$}\label{sec_W}
\setcounter{equation}{0}

In this section, we deduce new formulas for the wave operator $W_-$. Note that
similar work could be performed on $W_+$, but for simplicity we concentrate only
on one of the wave operators. 

By using the content of Lemma \ref{lem_statio} and the resolvent identities \eqref{eq_resolv_eq},
one easily infers  on suitable $f,g\in \H$ that
\begin{align*}
\langle W_- f,g\rangle
& =  \lim_{\varepsilon\searrow 0}\frac{\varepsilon}{\pi} \int_{\R}\d \lambda 
\langle R_0(\lambda- i \varepsilon)f, R^*(\lambda - i \varepsilon) g\rangle \\
& =  \lim_{\varepsilon\searrow 0}\frac{\varepsilon}{\pi} \int_{\R}\d \lambda 
\langle R_0(\lambda- i \varepsilon)f, R_0(\lambda - i \varepsilon) g\rangle \\
& \quad - \lim_{\varepsilon\searrow 0}\frac{\varepsilon}{\pi} \int_{\R}\d \lambda 
\langle  
R_0(\lambda+i\varepsilon) v\big(\overline{u} +v R_0(\lambda+i\varepsilon)v\big)^{-1}vR_0(\lambda+i\varepsilon) R_0(\lambda- i \varepsilon)f, g\rangle \\
& =  \langle f, g\rangle 
 - \lim_{\varepsilon\searrow 0} \int_{\R}\d \lambda 
\langle  
R_0(\lambda+i\varepsilon) v\big(\overline{u} +v R_0(\lambda+i\varepsilon)v\big)^{-1}v \delta_\varepsilon(H_0-\lambda)f, g\rangle 
\end{align*}
with
$
\delta_\varepsilon(H_0-\lambda)
:=\frac\varepsilon\pi R_0(\lambda\mp i\varepsilon)\;\!R_0(\lambda\pm i\varepsilon).
$

We now look at this expression in the spectral representation of $H_0$.
For this we set 
$$
\Hrond:=\ltwo(\R_+, \d \lambda;\HS)
$$ 
and write $L$ for the self-adjoint operator on $\Hrond$ corresponding to the multiplication by the variable $\lambda$.
Then for $\varphi, \psi \in \Hrond$ one has
\begin{align}\label{eq_WOP}
\nonumber &\langle \F_0 (W_--1)\F_0^*\varphi, \psi\rangle_\Hrond  \\
\nonumber & =  - \lim_{\varepsilon\searrow 0} \int_{\R}\d \lambda 
\langle  \F_0
R_0(\lambda+i\varepsilon) v\big(\overline{u} +v R_0(\lambda+i\varepsilon)v\big)^{-1}v \delta_\varepsilon(H_0-\lambda)\F_0^*\varphi, \psi\rangle_\Hrond \\
\nonumber & =  - \lim_{\varepsilon\searrow 0} \int_{\R}\d \lambda 
\langle  \F_0 v \big(\overline{u} +v R_0(\lambda+i\varepsilon)v\big)^{-1}v \F_0^*\delta_\varepsilon(L-\lambda)\varphi,  (L-\lambda+i\varepsilon)^{-1}\psi\rangle_\Hrond \\
\nonumber & =  - \lim_{\varepsilon\searrow 0} \int_{\R}\d \lambda \int_0^\infty \d \mu 
\big\langle  \F_0(\mu) v \big(\overline{u} +v R_0(\lambda+i\varepsilon)v\big)^{-1}v \F_0^* \delta_\varepsilon(L-\lambda)\varphi,  (\mu-\lambda+i\varepsilon)^{-1}\psi(\mu)\big\rangle_\HS \\
& =  - \lim_{\varepsilon\searrow 0} \int_{\R}\d \lambda \int_0^\infty \d \mu 
\big\langle  \lambda^{1/4}\big(\overline{u} +v R_0(\lambda+i\varepsilon)v\big)^{-1}v \F_0^*\delta_\varepsilon(L-\lambda)\varphi, 
 \frac{ \lambda^{-1/4}\mu^{1/4}}{\mu-\lambda+i\varepsilon} \mu^{-1/4}v\F_0(\mu)^*\psi(\mu)\big\rangle_\HS.
\end{align}

Motivated by the last expression, let us recall a few useful results.
From now on, assume that $\alpha >3$ in \eqref{eq_decay_V}. 
Let us also use the notation $C_{\rm c}(\R_+;\G)$ for the set of compactly
supported and continuous functions from $\R_+$ to some Hilbert space $\G$. With this
notation and what precedes, we note that the multiplication operator
$M(L) :C_{\rm c}(\R_+;\H)\to\Hrond$ given for $\xi\in C_{\rm c}(\R_+;\H)$ and $\lambda\in\R_+$
 by 
\begin{equation*}
(M(L)\xi)(\lambda)\equiv M(\lambda)\xi(\lambda):=\lambda^{-1/4}\F_0(\lambda)v\;\!\xi(\lambda)
\end{equation*}
extends to an element of
$\B\big(\ltwo(\R_+;\H),\Hrond\big)$.

The following statement will also be useful. It directly follows from \cite[Lem.~2.3]{RT13}.

\begin{Lemma}\label{lem_limite}
Let $\alpha >3$ in \eqref{eq_decay_V}, $\lambda\in\R_+$ and $\varphi\in C_{\rm c}(\R_+;\HS)$. Then,
we have
$$
\lim_{\varepsilon\searrow 0}\big\|v\F_0^*\;\!\delta_\varepsilon(L-\lambda)\varphi
-v\F_0(\lambda)^*\varphi(\lambda)\big\|=0.
$$
\end{Lemma}

We introduce one more multiplication operator. The following statement can be inferred from Lemma \ref{lm:properties-R} together with Lemma \ref{lem_Fpm}. 

\begin{Lemma}\label{lem_on_sigma}
 Suppose that $V:\R^3\to \C_-$ is a bounded, measurable function satisfying \eqref{eq_decay_V} with $\alpha>3$ and $\Im V<0$ on a non-trivial open set.
Then the function
$$
\R_+\ni\lambda\mapsto
\lambda^{1/4}\;\!\big(\overline{u} +v R_0(\lambda+i0)v\big)^{-1}v\F_0(\lambda)^*\in\B(\HS,\H)
$$
is continuous and bounded, and the multiplication operator
$B_0(L): C_{\rm c}\big(\R_+;\HS\big)\to\ltwo(\R_+;\H)$ given
for $\varphi\in C_{\rm c}\big(\R_+;\HS\big)$ and $\lambda\in\R_+$ by 
\begin{equation}\label{defdeB}
(B_0(L) \;\!\varphi)(\lambda) 
\equiv B_0(\lambda) \varphi(\lambda) :=\lambda^{1/4}\;\!\big(\overline{u} +v R_0(\lambda+i0)v\big)^{-1}v\F_0(\lambda)^*\varphi(\lambda) 
\end{equation}
extends to an element of $\B\big(\Hrond,\ltwo(\R_+;\H)\big)$.
\end{Lemma}

For $\varepsilon>0$ and $\lambda, \mu \in \R_+$ let us finally set
$$
\Theta_\varepsilon(\lambda,\mu):=\frac{1}{2\pi i} \frac{\lambda^{-1/4}\mu^{1/4}}{\mu-\lambda+i\varepsilon}
$$
and denote by $\Theta_\varepsilon$ the corresponding integral operator in $\ltwo(\R_+)$.
By using the notations introduced in the above lemma, one can rewrite \eqref{eq_WOP} in the following form
\begin{equation}\label{eq_W_initial}
\langle \F_0 (W_--1)\F_0^*\varphi, \psi\rangle_\Hrond 
=  -2\pi i \lim_{\varepsilon\searrow 0} \int_{\R}\d \lambda \int_0^\infty \d \mu 
\big\langle  B_\varepsilon(\lambda) \varphi, 
 \Theta_\varepsilon(\lambda,\mu)M(\mu)^* \psi(\mu)\big\rangle_\H,
\end{equation}
where
$$
B_\varepsilon(\lambda):=\lambda^{1/4}\;\!\big(\overline{u} +v R_0(\lambda+i0)v\big)^{-1}v\F_0^*
\delta_\varepsilon(L-\lambda).
$$

We can now finally state and prove the main result of this section.

\begin{Theorem}\label{thm_formula}
Suppose that $V:\R^3\to \C_-$ is a bounded, measurable function satisfying \eqref{eq_decay_V} with $\alpha>7$ and $\Im V<0$ on a non-trivial open set.
Then the following 
equalities hold:
\begin{equation}\label{eq_thm_1}
\F_0 (W_--1)\F_0^* = -2\pi i\;\!M(L)\;\!\big\{\tfrac12\big(1-\tanh(2\pi A_+)-i\cosh(2\pi A_+)^{-1}\big)\otimes1_{\H}\big\}B_0(L)
\end{equation}
with $A_+$ the generator of dilations in $\ltwo(\R_+)$, and  
\begin{equation*}
W_-=1+\tfrac12\big(1+\tanh(\pi A)-i\cosh(\pi A)^{-1}\big)(S-1)+K
\end{equation*}
with $A$ the generator of dilations in $\H$ and $K\in \K(\H)$.
\end{Theorem}

\begin{proof}
$(i)$ The proof of the first statement can be borrowed from the proof of Theorem 2.6 in \cite{RT13}, with
only minor modifications.
It consists in showing that the expression
\begin{equation}\label{eq_starting_p}
-  \int_{\R}\d \lambda \lim_{\varepsilon\searrow 0} \int_0^\infty \d \mu 
\big\langle B_\varepsilon(\lambda)\varphi, 
 \Theta_\varepsilon(\lambda,\mu)M(\mu)^*\psi(\mu)\big\rangle
\end{equation}
is well-defined for $\varphi$ and $\psi$ in dense subsets of $\Hrond$
and equal to
$\big\langle\;\!\F_0(W_\pm-1)\;\!\F_0^*\varphi,\psi\big\rangle_{\!\Hrond}$, 
once the constant $2\pi i$ is added.
Second, we show the equality \eqref{eq_thm_1}.

Take $\varphi\in C_{\rm c}(\R_+;\HS)$ and
$\psi\in C_{\rm c}^\infty(\R_+)\odot C(\S^2)$, then we
have for each $\varepsilon>0$, $\lambda\in\R_+$ and $\mu\in\R_+$ the inclusions
$$
g_\varepsilon(\lambda):=\lambda^{1/4}\big(\overline{u} +v R_0(\lambda+i\varepsilon)v\big)^{-1}v \F_0^*\delta_\varepsilon(L-\lambda)\varphi \in \H
\quad\hbox{and}\quad
f(\mu):=\mu^{-1/4}v\F_0(\mu)^*\psi(\mu)\in\H\;\!.
$$
The fact that $f(\mu)\in\H$ follows from Lemma \ref{lem_Fpm}. With these notations the expression \eqref{eq_starting_p} is equal to
\begin{equation*}
-\int_{\R_+}\d\lambda\,\lim_{\varepsilon\searrow0}\int_0^\infty\d\mu\,
\Big\langle g_\varepsilon(\lambda),\frac{\lambda^{-1/4}\mu^{1/4}}
{\mu-\lambda+i\varepsilon}\;\!f(\mu)\Big\rangle.
\end{equation*}

Now, using the formula
$
(\mu-\lambda+i\varepsilon)^{-1}
=-i\int_0^\infty\d z\e^{i(\mu-\lambda)z}\e^{-\varepsilon z}
$
and then applying Fubini's theorem, one obtains that
\begin{align}
&\lim_{\varepsilon\searrow0}\int_0^\infty\d\mu\,\Big\langle g_\varepsilon(\lambda),
\frac{\lambda^{-1/4}\mu^{1/4}}{\mu-\lambda+i\varepsilon}\;\!f(\mu)
\Big\rangle\nonumber\\
&=-i\lim_{\varepsilon\searrow0}\int_0^\infty\d z\,\e^{-\varepsilon z}
\Big\langle g_\varepsilon(\lambda),\int_{-\lambda}^\infty\d\nu\,\e^{i\nu z}
\left(\frac{\nu+\lambda}{\lambda}\right)^{1/4}f(\nu+\lambda)
\Big\rangle.\label{eq21}
\end{align}
Furthermore, the integrant in \eqref{eq21} can be bounded independently of
$\varepsilon\in(0,1)$. Indeed, one has
\begin{align}\label{eq_dure}
\begin{split}
&\left|\,\e^{-\varepsilon z}\Big\langle g_\varepsilon(\lambda),
\int_{-\lambda}^\infty\d\nu\,\e^{i\nu z}
\left(\frac{\nu+\lambda}{\lambda}\right)^{1/4}f(\nu+\lambda)
\Big\rangle\right|\\
&\le\big\|g_\varepsilon(\lambda)\big\|\;\!\bigg\|\int_{-\lambda}^\infty\d\nu\,
\e^{i\nu z}\left(\frac{\nu+\lambda}{\lambda}\right)^{1/4}
f(\nu+\lambda)\bigg\|,
\end{split}
\end{align}
and we know from Lemmas \ref{lm:properties-R} and \ref{lem_limite} that
$g_\varepsilon(\lambda)$ converges to
$
B_0(\lambda)\varphi(\lambda)
$
in $\H$ as $\varepsilon\searrow0$, with $B_0(\lambda)$ defined in \eqref{defdeB}. 
Therefore, the family
$\|g_\varepsilon(\lambda)\|$ (and thus the r.h.s. of \eqref{eq_dure}) is
bounded by a constant independent of $\varepsilon\in(0,1)$.

In order to exchange the integral over $z$ and the limit $\varepsilon\searrow0$ in
\eqref{eq21}, it remains to show that the second factor in \eqref{eq_dure} belongs to
$\lone(\R_+,\d z)$. For that purpose, we denote by $h_\lambda$ the trivial extension
of the function
$
(-\lambda,\infty)\ni\nu\mapsto
\big(\frac{\nu+\lambda}{\lambda}\big)^{1/4}f(\nu+\lambda)\in\H
$
to all of $\R$, and then note that the second factor in \eqref{eq_dure} can be rewritten
as $(2\pi)^{1/2}\|(\F_1^*h_\lambda)(z)\|$, with $\F_1$ the one-dimensional
Fourier transform. To estimate this factor, observe that if $D_1$ denotes the
self-adjoint operator $-i\nabla$ on $\R$, then
$$
\big\|\big(\F_1^*h_\lambda\big)(z)\big\|
=\langle z\rangle^{-2}
\big\|\big(\F_1^*\langle D_1\rangle^2h_\lambda\big)(z)\big\|,
\quad z\in\R_+\;\!.
$$
Consequently, one would have that
$
\|(\F_1^*h_\lambda)(z)\|\in\lone(\R_+,\d z)
$
if the norm
$
\big\|\big(\F_1^*\langle D_1\rangle^2h_\lambda\big)(z)\big\|
$
were bounded independently of $z$. Now, if $\psi=\eta\otimes\xi$ with
$\eta\in C_{\rm c}^\infty(\R_+)$ and $\xi\in C(\S^2)$, then one has for any
$x\in\R^3$
$$
\big(f(\nu+\lambda)\big)(x)
=\frac1{4\pi^{3/2}}\;\!\eta(\nu+\lambda)v(x)\int_{\S^2}\d\omega\,
\e^{i\sqrt{\nu+\lambda}\;\!\omega\cdot x}\xi(\omega).
$$
Therefore, one has
\begin{equation*}
\big(h_\lambda(\nu)\big)(x)=
\begin{cases}
\frac1{4\pi^{3/2}}\left(\frac{\nu+\lambda}{\lambda}\right)^{1/4}\eta(\nu+\lambda) v(x)
\int_{\S^2}\d\omega\,\e^{i\sqrt{\nu+\lambda}\;\!\omega\cdot x}\xi(\omega)
& \nu>-\lambda\\
0 & \nu\le-\lambda,
\end{cases}
\end{equation*}
which in turns implies that
$$
\big|\big\{\big(\F_1^*\langle D_1\rangle^2h_\lambda\big)(z)\big\}(x)\big|
\le{\rm Const.}\;\!v(x)\;\!\langle x\rangle^2,
$$
with a constant independent of $x\in\R^3$ and $z\in\R_+$. Since the r.h.s. belongs to
$\H$, one concludes that
$\big\|\big(\F_1^*\langle D_1\rangle^2h_\lambda\big)(z)\big\|_{\H}$ is bounded
independently of $z$ for each $\psi=\eta\otimes\xi$, and thus for each
$\psi\in C^\infty_{\rm c}(\R_+)\odot C(\S^2)$ by linearity. As a consequence, one can
apply Lebesgue dominated convergence theorem and obtain that \eqref{eq21} is equal to
$$
-i\;\!\bigg\langle g_0(\lambda),\int_0^\infty\d z\int_\R\d\nu\;\!
\e^{i\nu z}h_\lambda(\nu)\bigg\rangle.
$$

With this equality, one can conclude the first part of the proof for the first statement.
Indeed, one has
shown that \eqref{eq_starting_p} is well defined on the dense set of vectors introduced at the
beginning of the proof. Then, its equality with \eqref{eq_W_initial} (modulo the constant)
follows by a standard argument as presented for example in \cite[Lem.~5.2.2]{Yaf92}. 

$(ii)$ Let us now show that 
$
\big\langle\;\!\F_0(W_\pm-1)\;\!\F_0^*\varphi,\psi\big\rangle_{\!\Hrond}
$
is equal to
$
\big\langle-2\pi i\;\!M(L)\big\{\vartheta(A_+)\otimes1_{\H}\big\}B(L)\varphi,
\psi\big\rangle_{\!\Hrond}
$
with 
\begin{equation}\label{defvar}
\vartheta(\nu):=\frac12\big(1-\tanh(2\pi\nu)-i\cosh(2\pi\nu)^{-1}\big),\quad\nu\in\R.
\end{equation}
For that purpose, we write $\chi_+$ for the characteristic function for $\R_+$. Since
$h_\lambda$ has compact support, we obtain the following equalities in the sense of
distributions (with values in $\H$)\;\!:
\begin{align*}
\int_0^\infty\d z\int_\R\d\nu\;\!\e^{i\nu z}h_\lambda(\nu)
&=\sqrt{2\pi}\int_\R\d\nu\,\big(\F_1^*\chi_+\big)(\nu)\;\!h_\lambda(\nu)\\
&=\sqrt{2\pi}\int_{-\lambda}^\infty\d\nu\,\big(\F_1^*\chi_+\big)(\nu)
\left(\frac{\nu+\lambda}{\lambda}\right)^{1/4}f(\nu+\lambda)\\
&=\sqrt{2\pi}\int_\R\d\mu\,\big(\F_1^*\chi_+\big)\big(\lambda(\e^\mu-1)\big)
\;\!\lambda\e^{5\mu/4}f(\e^\mu\lambda)\qquad(\e^\mu\lambda:=\nu+\lambda)\\
&=\sqrt{2\pi}\int_\R\d\mu\,\big(\F_1^*\chi_+\big)\big(\lambda(\e^\mu-1)\big)
\;\!\lambda\e^{3\mu/4}\big\{\big(U_\mu^+\otimes1_{\H}\big)f\big\}(\lambda).
\end{align*}
Then, by using the fact that
$
\F_1^*\chi_+
=\sqrt{\frac\pi2}\;\!\delta_0+\frac i{\sqrt{2\pi}}\;\!\Pv\frac1{(\;\!\cdot\;\!)}
$
with $\delta_0$ the Dirac delta distribution and $\Pv$ the principal value, one gets
that
$$
\int_0^\infty\d z\int_\R\d\nu\;\!\e^{i\nu z}h_\lambda(\nu)
=\int_\R\d\mu\left(\pi\;\!\delta_0(\e^\mu-1)
+i\;\!\Pv\frac{\e^{3\mu/4}}{\e^\mu-1}\right)
\big\{\big(U_\mu^+\otimes1_{\H}\big)f\big\}(\lambda).
$$
So, by considering the identity
$$
\frac{\e^{3\mu/4}}{\e^\mu-1}
=\frac14\left(\frac1{\sinh(\mu/4)}+\frac1{\cosh(\mu/4)}\right)
$$
and the equality \cite[Table 20.1]{Jef95}
\begin{equation*}
\big(\F_1\bar\vartheta\big)(\nu)
=\sqrt{\frac\pi2}\;\!\delta_0\big(\e^\nu-1\big)
+\frac i{4\sqrt{2\pi}}\;\!\Pv\left(\frac1{\sinh(\nu/4)}+\frac1{\cosh(\nu/4)}\right),
\end{equation*}
with $\vartheta$ defined in \eqref{defvar}, one infers that
\begin{align*}
&\big\langle\F_0(W_--1)\;\!\F_0^*\varphi,\psi\big\rangle_{\!\Hrond}\\
&=i\int_{\R_+}\d\lambda\,\bigg\langle g_0(\lambda),\int_\R\d\mu\,
\bigg\{\pi\;\!\delta_0\big(\e^\mu-1\big) +\frac i4\;\!\Pv\left(\frac1{\sinh(\mu/4)}
+\frac1{\cosh(\mu/4)}\right)\bigg\}
\big\{\big(U_\mu^+\otimes1_{\H}\big)f\big\}(\lambda)
\bigg\rangle\\
&=i\sqrt{2\pi}\int_{\R_+}\d\lambda\left\langle g_0(\lambda),\int_\R\d\mu\,
\big(\F_1\bar\vartheta\big)(\mu)\;\!
\big\{\big(U_\mu^+\otimes1_{\H}\big)f\big\}(\lambda)
\right\rangle.
\end{align*}
Finally, by recalling that
$
\big\{\vartheta(A_+)\otimes1_{\H}\big\}f
=\frac1{\sqrt{2\pi}}\int_\R\d\mu\,\big(\F_1\bar\vartheta\big)(\mu)
\big(U_\mu^+\otimes 1_{\H}\big)f
$,
that $g_0(\lambda)=(B(L)\varphi)(\lambda)$ and that $f=M(L)^*\psi$, one obtains
\begin{align*}
\big\langle\F_0(W_--1)\F_0^*\varphi,\psi\big\rangle_{\!\Hrond}
&=2\pi i\int_{\R_+}\d\lambda\,\big\langle(B(L)\varphi)(\lambda),
\big\{\big(\vartheta(A_+)^*\otimes1_{\H}\big)M(L)^*\psi\big\}(\lambda)
\big\rangle\\
&=\big\langle-2\pi i\;\!M(L)\;\!\big\{\vartheta(A_+)\otimes1_{\H}\big\}B(L)\varphi,
\psi\big\rangle_{\!\Hrond}.
\end{align*}
This concludes the proof, since the sets of vectors $\varphi\in C_{\rm c}(\R_+;\HS)$
and $\psi\in C_{\rm c}^\infty(\R_+)\odot C(\S^2)$ are dense in $\Hrond$.

$(iii)$ The proof of the second statement can mimicked from the same proof in \cite{RT13}.
The main arguments are:
\begin{enumerate}[label=(\roman*)]
\item The fact that $\big\{\vartheta(A_+)\otimes1_{\HS}\big\}M(L)
-M(L)\big\{\vartheta(A_+)\otimes 1_{\H}\big\}$
belongs to $\K\big(\ltwo(\R_+;\H),\Hrond\big)$, as shown in \cite[Lem.~2.7]{RT13},
\item The equality
$\F_0 \;\!\frac12\big(1+\tanh(\pi A)-i\cosh(\pi A)^{-1}\big) \!\;\F_0^*=\vartheta(A_+)\otimes 1_{\HS},$,
which can be checked by a direct computation,  
\item The equality
$-2\pi iM(\lambda)B_0(\lambda)=S(\lambda)-1$ which can be inferred from \eqref{eq_S_stationary}.
\end{enumerate}
By using successively these three arguments, one directly deduces the second statement from the first one.
\end{proof}

\section{Levinson's theorem}\label{sec_topol}
\setcounter{equation}{0}

Motivated by the expressions obtained in the previous section, let us introduce
a $C^*$-algebraic framework.
Note that since Theorem \ref{thm_formula} holds under the conditions $\alpha > 7 $ in \eqref{eq_decay_V} and $\Im V<0$ on a non-trivial open set, we shall assume these assumptions throughout this section. Clearly, some of the statements could be obtained under weaker conditions.

Recall that $A$ denotes the generator of dilations in $\H$ while $H_0$ stands for the Laplace
operator. If we set $B:=\frac{1}{2}\ln(H_0)$, then the operators $A$ and $B$ satisfy the important relation $[iB,A]=1$, or more precisely
$$
\e^{itB}A\e^{-itB}=A+t, \qquad \e^{isA}B\e^{-isA}=B-s.
$$
We also recall that the spectrum of $A$ is $\R$, and that the spectrum of $H_0$ is $[0,\infty)$.
We then set
$$
\EE:=C^*\Big(\eta(A)\psi(H_0)\mid \eta\in C\big([-\infty, \infty];\C\big) \hbox{ and } 
\psi\in C\big([0,\infty);\K(\HS)\big)\Big)^+
$$
where the exponent $+$ means that $\C$ times the identity has been added to the algebra.
In the context of Schr\"o\-dinger operators in $\R^3$ this $C^*$-algebra has already been introduced 
and studied in \cite[Sec.~4]{KR12}.

Let us state a few results about this algebra. It is known that $\K(\H)\subset \EE$, and
since $\K(\H)$ is an ideal in $\EE$ the quotient algebra can be computed. One has
$$
\EE/\K(\H) \cong C\big(\square; \K(\HS)\big)^+,
$$
where $\square$ denotes the edges of a square. More precisely, if we set 
$$
q:\EE \to  C\big(\square; \K(\HS)\big)^+
$$
for the quotient map, and if we consider a typical element of $\EE$ of the form $\eta(A)\psi(H_0)$
one has 
$$
q\big(\eta(A)\psi(H_0)\big) = \Big(\eta(\cdot)\psi(0), \eta(+\infty)\psi(\cdot), \eta(\cdot)\psi(+\infty), \eta(-\infty)\psi(\cdot)\Big)
$$
where (by convention) we started by the edge on the left of the square and list  the elements clockwise.
Note also that continuity holds at the four corners, as it can easily be checked.

Since the three algebras can be fitted in the short-exact sequence
$$
0 \to \K(\H)\to \EE\to C\big(\square; \K(\HS)\big)^+ \to 0,
$$
the index map in $K$-theory maps the $K_1$-group of $C\big(\square; \K(\HS)\big)^+$
to the $K_0$-group of $\K(\H)$. 
It is well-known that both groups can be identified with $\Z$, with the usual trace for the former.
For the latter, the winding number of the pointwise determinant induces the isomorphism.
However, a careful use of this functional is necessary, as it will appear later in the application.
  
Our interest in the algebra $\EE$ is that the wave operator $W_-$ belongs to it, as a
consequence of Theorem \ref{thm_formula}. Indeed, the scattering operator can be thought as an operator
of the form $S(H_0)$ with a function 
$$
[0,\infty)\ni \lambda \mapsto S(\lambda)\in \K(\HS)
$$
continuous and having limits at $0$ and at $+\infty$.
It then follows that the image of $W_-$ in the quotient algebra can be computed and one infers that
\begin{equation}\label{eq_quot}
q(W_-) = \big(1, S(\cdot), 1, 1\big)
\end{equation}
where the facts that $S(0)=1$ and $\lim_{\lambda\to \infty}S(\lambda)=1$ (see Lemma \ref{lm:S(lambda)}) have been taken into account.
If in addition $S(\cdot)$ is invertible, we can deduce two complementary sets of information: 
\begin{enumerate}[label=(\roman*)]
\item Since the element $q(W_-)$ computed in \eqref{eq_quot} is invertible, it follows that the operator $W_-$ is a Fredholm operator (Atkinson's theorem, see for example \cite[Prop.~3.3.11]{Ped}) and that its range is closed.
Let us then define the orthogonal projection $P_-$ by
\begin{equation}\label{eq_P_-}
P_-\H = \Ran(W_-)^\bot.
\end{equation}
Under our current set of assumptions (recalled below) this projection is finite dimensional.
\item The element $q(W_-)$ defines an element $[q(W_-)]_1$ of the $K_1$-group of the $C^*$-algebra $C\big(\square; \K(\HS)\big)^+$.
Consequently, it follows that this element is related to the Fredholm index of $W_-$ by the relation
\begin{equation}\label{eq_topo_form}
-\dim(P_-)=\index(W_-) = \big(K_0(\Tr)\circ \delta_1 \big)([q(W_-)]_1),
\end{equation}
where the first equality follows from the injectivity of $W_-$ given by Lemma \ref{lem:existence_wave}, and where $\delta_1$ stands for the index map in $K$-theory.
We refer to \cite[Prop.~9.4.2]{RLL} for the second equality and for more explanations. 
\end{enumerate}

Our final aim is to get an analytic formula for the computation of \eqref{eq_topo_form}. 
This will be evaluated with an expression of the form
$$
\frac{1}{2\pi i} \int_0^\infty \d \lambda \;\!\left(\Tr\big(S(\lambda)^{-1}S'(\lambda)\big) + c \lambda^{-\frac12} \right)
$$
for a suitable constant $c$, and whenever $S(\lambda)^{-1}$ makes sense.
Here and in the sequel, any integrals of this type should be understood as improper Riemann integrals.
Note that previous investigations of this type made use of the spectral shift function for the pair $(H,H_0)$. Since $H$ is not self-adjoint, this tool seems no longer available to us. Nevertheless, enough of the elements of the proof can still be used to construct an integral formula for the index.

First we analyze some properties of the determinant of the scattering matrix. 
Recall from Lemma \ref{lm:S(lambda)} that $S(\lambda)-1 \in \L^1(\HS)$ for all $\lambda \in [0,\infty)$ and that the map $\lambda \mapsto S(\lambda)-1$ is continuously differentiable in $\L^1(\HS)$, with $\L^1(\HS)$ the set of trace-class operator on the Hilbert space $\HS$.
Thus we infer from \cite[Equation IV.1.14]{GK} that
\begin{equation*}
    \frac{\d}{\d \lambda} \ln \det\big(S(\lambda)\big) = \Tr\big(S(\lambda)^{-1} S'(\lambda)\big)
\end{equation*}
for all $\lambda \in (0,\infty)$, whenever $S(\lambda)$ is invertible.

The following is another straightforward adaptation of \cite[Prop.~9.1.1, 9.1.2 and 9.1.3]{Yaf10}.
For its statement we use the notation
$$
D_2(z) := {\det}_2\big(1+uvR_0(z)v\big).
$$

\begin{Proposition}\label{prop-det2}
Suppose that $V:\R^3\to \C_-$ is a bounded, measurable function satisfying \eqref{eq_decay_V} with $\alpha>7$ and $\Im V<0$ on a non-trivial open set. 
Assume also that the operator $\overline{u} + vR_0(\lambda-i0)v$ is invertible in $\B(\H)$ for any $\lambda>0$.
Then the function $z\mapsto D_2(z)$ is analytic in $\C \setminus [0,\infty)$. In addition, this function is continuous for $z \in \C \setminus [0,\infty)$ up to the cut along $[0,\infty)$, with the point $z = 0$ possibly excluded. For $\lambda \in (0,\infty)$ the function $\lambda \mapsto D_2(\lambda \pm i 0)$ does not have any zeros. Furthermore, we have
\begin{equation*}
\lim_{|z| \to \infty} D_2(z) = 1
\end{equation*}
uniformly in $\arg(z)$.
\end{Proposition}

Since we suppose that $\alpha > 7$ in \eqref{eq_decay_V}, $V$ is clearly integrable. We then define the constant  
\begin{equation}\label{eq_constant}
c := \frac{2\pi i \vol(\S^2)}{4(2\pi)^3} \int_{\R^3} \d x \;\!V(x)
= \frac{i}{4 \pi} \int_{\R^3} \d x\;\! V(x).
\end{equation}

Our first task is to show the integrability of the map $\lambda \mapsto \frac{\d}{\d \lambda} \ln\det\big(S(\lambda)\big) +c \lambda^{-\frac12}$.
Since $V$ is not real-valued, this result is not available in the literature.  

\begin{Lemma}\label{lem:high-energy-limit}
Suppose that $V:\R^3\to \C_-$ is a bounded, measurable function satisfying \eqref{eq_decay_V} with $\alpha>7$ and $\Im V<0$ on a non-trivial open set.
Assume also that the operator $\overline{u} + vR_0(\lambda-i0)v$
is invertible in $\B(\H)$ for any $\lambda > 0$.
Then the function $(0,\infty) \ni \lambda \mapsto  \ln \det\big(S(\lambda)\big)$ is in $C^1\big((0,\infty)\big)$ and 
\begin{align}\label{eq:high-energy-limit}
\lim_{\lambda \to \infty} \frac{1}{2\pi i}\left( \ln \det\big(S(\lambda)\big) + 2 c \lambda^\frac12 \right) &= m
\end{align}
for some $m \in \Z$.
\end{Lemma}

\begin{proof}
The proof of  the differentiability claim can be directly copied from the proof of \cite[Thm.~9.1.18]{Yaf10}, using Lemmas \ref{lm:S(lambda)} and \ref{lem_S_invert}.
For the second statement, consider $z, \overline{z} \in \C \setminus \sigma(H)$, 
and observe that the following identity holds:
\begin{align*}
(1+uvR_0(z)v)^{-1}(1+uvR_0(\overline{z})v)  &=(\overline{u}+vR_0(z)v)^{-1}(\overline{u}+vR_0(\overline{z})v) \\ 
&= 1- (\overline{u}+vR_0(z)v)^{-1} v \big(R_0(z)-R_0(\overline{z})\big)v.
\end{align*}
By Lemma \ref{lm:properties-R} we have that the above function is continuous and can be extended to $z = \lambda \pm i0$ for $\lambda > 0$. Recall the relation $2\pi i \F_0(\lambda)^*\F_0(\lambda) =  \big(R_0(\lambda+i0)-R_0(\lambda-i0)\big)$. Then we find, using also the relation $\det(1+AB) = \det(1+BA)$, that 
\begin{align*}
\det\big(S(\lambda)\big) &= \det\big(1-2\pi i \F_0(\lambda) v \big(\overline{u}+vR_0(\lambda+i0)v\big)^{-1} v \F_0(\lambda)^*\big) \\
&= \det\big((1+uvR_0(\lambda+i0)v)^{-1}(1+uvR_0(\lambda-i0)v)\big) \\
&= \frac{\det_2(1+uvR_0(\lambda-i0)v)}{\det_2(1+uvR_0(\lambda+i0)v)} \e^{-\Tr(uv(R_0(\lambda+i0)-R_0(\lambda-i0))v)},
\end{align*}
where in the last line we have used \cite[Lem.~III.4]{guillope-thesis}. Since $D_2(\lambda \pm i0) \to 1$ as $\lambda \to \infty$ by Proposition \ref{prop-det2}, the high-energy behavior of $\det\big(S(\lambda)\big)$ is fully determined by the term $\e^{-\Tr(uv(R_0(\lambda+i0)-R_0(\lambda-i0))v)}$. We can compute this trace explicitly, namely
\begin{align}\label{eq:trace-computation}
\Tr\Big(uv\big(R_0(\lambda+i0)-R_0(\lambda-i0)\big)v\Big) &= \frac{2\pi i \lambda^{\frac12} \vol(\S^2)}{2(2\pi)^3} \int_{\R^3}\d x\;\! V(x)  = 2 c \lambda^\frac12,
\end{align}
leading easily to the statement.
\end{proof}

\begin{Remark}\label{rmk:branch}
We can fix a branch of the meromorphic function $\ln \det_2\big(1+uvR_0(z)v\big)$ by the condition
\begin{align*}
\lim_{|z| \to \infty} \arg\Big({\det}_2\big(1+uvR_0(z)v\big) \Big)&= 0.
\end{align*}
Fixing this branch gives $m = 0$ in Lemma \ref{lem:high-energy-limit} and we take this convention for the rest of this paper.
\end{Remark}

\begin{Lemma}\label{lem:integrable}
Under the assumptions of Lemma \ref{lem:high-energy-limit}, for any $\Lambda > 0$ the map $\lambda \mapsto \frac{\d}{\d \lambda} \ln\det\big(S(\lambda)\big) + c \lambda^{-\frac12}$ is integrable on $(0,\Lambda)$ and $\int_0^\Lambda \d\lambda \big(\frac{\d}{\d \lambda} \ln\det\big(S(\lambda)\big) + c \lambda^{-\frac12}\big)$, with $c$ defined in \eqref{eq_constant}, converges as $\Lambda\to\infty$.
\end{Lemma}

The following proof is an analogue of \cite[Lem.~4.12]{Angus} valid for usual self-ajdoint Schr\"{o}dinger operators.

\begin{proof}
We first check integrability in a neighbourhood of zero. Recall that we have the equality 
$$
\frac{\d}{\d \lambda} \ln \det\big(S(\lambda)\big) = \Tr\big(S(\lambda)^{-1} S'(\lambda)\big)
$$ 
for $\lambda \in (0,\infty)$. 
Recall also that
\begin{align*}
S(\lambda) &= 1-2\pi i \F_0(\lambda) v \big(\overline{u} +v R_0(\lambda+i0)v\big)^{-1}v \F_0(\lambda)^*.
\end{align*}
The operator $\F_0(\lambda)v\in \B(\H,\HS)$ has Hilbert-Schmidt norm
\begin{align*}
\|\F_0(\lambda) v \|_2^2 &= (2\pi)^{-3} \lambda^{\frac12} \int_{\R^3} \int_{\S^2} |v(x)|^2 \, \d \omega \, \d x = C \lambda^{\frac{1}{4}},
\end{align*}
with the same calculation applying to the adjoint also, from which we obtain
\begin{align*}
\|\F_0(\lambda)v \|_2 &\leq C \lambda^{\frac{1}{4}}, \quad \hbox{ and } \quad \|v\F_0(\lambda)^*\|_2 \leq C \lambda^{\frac{1}{4}}.
\end{align*}
A similar calculation shows that
\begin{align*}
& \|\F_0'(\lambda) v \|_2^2 \\
&\leq \frac{1}{4} (2\pi)^{-3} \lambda^{-\frac{3}{2}} \int_{\R^3} \int_{\S^2} |v(x)|^2 \, \d \omega \, \d x + \frac{1}{4} (2\pi)^{-3} \lambda^{-\frac12} \int_{\R^3} \int_{\S^2} |\langle x, \omega \rangle|^2 |v(x)|^2 \, \d \omega \, \d x  \\
&= C_1 \lambda^{-\frac{3}{2}} + C_2 \lambda^{-\frac12}.
\end{align*}
Thus for sufficiently small $\lambda$ we have the estimates
\begin{align*}
\|\F_0'(\lambda)v \|_2 &\leq C \lambda^{-\frac{3}{4}}, \quad \hbox{ and } \quad \|v\F_0'(\lambda)^*\|_2 \leq C \lambda^{-\frac{3}{4}}.
\end{align*}

Then we can compute that 
\begin{align*}
 S'(\lambda) 
&= -2\pi i \Big(\F_0'(\lambda)v \big(\overline{u} +v R_0(\lambda+i0)v\big)^{-1}v \F_0(\lambda)^* +\F_0(\lambda)v \big(\overline{u} +v R_0(\lambda+i0)v\big)^{-1}v \F_0'(\lambda)^* \\
&\quad +\F_0(\lambda)v \big(\overline{u} +v R_0(\lambda+i0)v\big)^{-1}v R_0(\lambda+i0)^2 
v \big(\overline{u} +v R_0(\lambda+i0)v\big)^{-1}v \F_0(\lambda)^* \Big).
\end{align*}
For sufficiently small $\lambda$, the operators 
$$
\big(\overline{u} +v R_0(\lambda+i0)v\big)^{-1} \quad \hbox{ and } \quad
\big(\overline{u} +v R_0(\lambda+i0)v\big)^{-1}v R_0(\lambda+i0)^2 v
\big(\overline{u} +v R_0(\lambda+i0)v\big)^{-1}
$$ 
are uniformly bounded (see Lemma \ref{lm:properties-R} for the resolvents, and use the well-known fact that $vR_0(\lambda+i0)^2v$ is bounded if $\alpha>2$ in \eqref{eq_decay_V} which follows from the explicit expression of the kernel of $R_0(\lambda+i0)$). Then we can use H\"{o}lder's inequality for Schatten norms to make the estimate 
\begin{align*}
&\big\|S'(\lambda)\big\|_1  \\
& \leq 2\pi \big\|\F_0'(\lambda)v \big(\overline{u} +v R_0(\lambda+i0)v\big)^{-1}v \F_0(\lambda)^* \big\|_1 + 2\pi \big\|\F_0(\lambda)v \big(\overline{u} +v R_0(\lambda+i0)v\big)^{-1}v \F_0'(\lambda)^* \big\|_1 \\
&+2\pi \big\|\F_0(\lambda)v\big(\overline{u} +v R_0(\lambda+i0)v\big)^{-1}v R_0(\lambda+i0)^2 v\big(\overline{u} +v R_0(\lambda+i0)v\big)^{-1}v \F_0(\lambda)^* \big\|_1 \\
&\leq 2\pi \|\F_0'(\lambda)v\|_2 \big\|\big(\overline{u} +v R_0(\lambda+i0)v\big)^{-1}\big\|_{\B(\H)} \|v \F_0(\lambda)^*\|_2 \\
& +2\pi  \|\F_0(\lambda)v\|_2  \big\|\big(\overline{u} +v R_0(\lambda+i0)v\big)^{-1}\big\|_{\B(\H)} \|v \F_0'(\lambda)^* \|_2\\
&+2\pi \big\|\F_0(\lambda)v\|_2  \big\|\big(\overline{u} +v R_0(\lambda+i0)v\big)^{-1}v R_0(\lambda+i0)^2 v\big(\overline{u} +v R_0(\lambda+i0)v\big)^{-1}\big\|_\infty \|v \F_0(\lambda)^* \|_2 \\
&\leq C_1 \lambda^{-\frac12}+C_2 \lambda^{\frac12},
\end{align*}
for some positive constant $C_1,C_2$. It remains to recall that $\lambda\mapsto S(\lambda)^{-1}$ is bounded on $[0,\infty)$ in the operator norm by Lemma \ref{lem_S_invert}. Thus we find using again H\"{o}lder's inequality for Schatten norms that for sufficiently small $\lambda$ we have the estimate 
\begin{align*}
\big|\Tr\big(S(\lambda)^{-1}S'(\lambda)\big)\big| &\leq \|S(\lambda)^{-1}S'(\lambda)\|_1 \leq \|S(\lambda)^{-1}\|_{\B(\HS)} \|S'(\lambda)\|_1 \leq C \lambda^{-\frac12}
\end{align*}
for some $C>0$, which is integrable near zero.

By Lemma \ref{lem:high-energy-limit}, for any $\Lambda > 0$ the map $(0,\Lambda) \ni \lambda \mapsto \frac{\d}{\d \lambda} \ln\det\big(S(\lambda)\big) + c \lambda^{-\frac12}$ is continuous (and in fact the derivative of a continuous function) and so by the Fundamental Theorem of Calculus,
\begin{equation}\label{eq:lim}
\int_0^\Lambda \d \lambda \left(\frac{\d}{\d \lambda} \ln\det\big(S(\lambda)\big) + c \lambda^{-\frac12} \right) < \infty.
\end{equation}
The limit of Equation \eqref{eq:high-energy-limit} shows that the integral in Equation \eqref{eq:lim} converges as $\Lambda \to \infty$.
\end{proof}

Since the map $\lambda \mapsto \det\big(S(\lambda)\big)$ does not necessarily converge as $\lambda \to \infty$, some regularisation is required. For $\lambda \in [0,\infty)$, define $A(\lambda) \in \B(\HS)$ by
\begin{equation*}
A(\lambda) := 4 \tan^{-1}(\lambda^\frac12) \F_0(\lambda) V \F_0(\lambda)^*.
\end{equation*}
This operator is trace class, by assuming $\alpha > 3$ in \eqref{eq_decay_V}.
Indeed, based on the definition of $\F_0(\lambda)$, we can write explicitly the integral kernel of $A(\lambda)$ as
\begin{align*}
A(\lambda,\theta,\theta') &= 2 \tan^{-1}(\lambda^\frac12 ) \lambda^\frac12 (2\pi)^{-3}  \int_{\R^3}  
\d x\;\!\e^{-i \lambda^\frac12 \langle \theta - \theta', x \rangle } V(x)  .
\end{align*}
Integrating along the diagonal we find
\begin{align*}
\Tr\big(A(\lambda)\big) &= \frac{2 \tan^{-1}(\lambda^\frac12) \lambda^\frac12 \vol(\S^2)}{(2\pi)^3} \int_{\R^3} \d x\;\!V(x)  
= \frac{4}{\pi i} \tan^{-1}(\lambda^\frac12) \lambda^\frac12 c.
\end{align*}
In addition, by the properties of the map $\lambda \mapsto \F_0(\lambda)$  already mentioned in Lemma \ref{lem_Fpm} we infer that the map $\lambda \mapsto A(\lambda) \in \B(\HS)$ is continuous and has norm limits 
$$
\ulim_{\lambda \to 0} A(\lambda) = \ulim_{\lambda \to \infty} A(\lambda) = 0. 
$$
As a consequence, we also see that the map $\lambda \mapsto \Tr\big(A(\lambda)\big)$ is continuous with $\lim_{\lambda \to 0} \Tr\big(A(\lambda)\big) = 0$ and 
\begin{align}\label{eq:det-limits}
\lim_{\lambda \to \infty} \left( i \Tr\big(A(\lambda)\big) - \frac{2\pi i \lambda^{\frac12} \vol(\S^2)}{2(2\pi)^3} \int_{\R^3} \d x\;\!V(x)  \right) &= \lim_{\lambda \to \infty} \left( i \Tr\big(A(\lambda)\big) - 2 c \lambda^\frac12 \right) = 0.
\end{align}

Based on these observations, we now define for $\lambda \in [0,\infty)$ the operator
\begin{align*}
\beta(\lambda) &= \e^{i A(\lambda)} \in \B(\HS),
\end{align*}
which satisfies $\det\big(\beta(\lambda)\big) = \e^{i \Tr(A(\lambda))}$ for all $\lambda \in [0,\infty)$. 
By construction we have the norm limits 
$$
\ulim_{\lambda \to 0} \beta(\lambda) = \ulim_{\lambda \to \infty} \beta(\lambda) = 1.
$$
The operator $\beta(\lambda)$ is invertible for all $\lambda$ with $\beta(\lambda)^{-1} = \e^{-i A(\lambda)}$. Such results are not true in the trace norm. We also define the operator $W_\beta \in \B(\Hrond)$ by the equality
\begin{equation*}
W_\beta = 1+ \big\{\tfrac12\big(1-\tanh(2\pi A_+)-i\cosh(2\pi A_+)^{-1}\big)\otimes1_{\H}\big\}(\beta(L)-1).
\end{equation*}

\begin{Lemma}\label{lem_Wb}
Under the assumptions of Lemma \ref{lem:high-energy-limit}, the operator $W_\beta$ is a Fredholm operator in $\B(\HS)$ 
satisfying $\index(W_\beta) = 0$.
\end{Lemma}

A similar statement has been proved in \cite[Lem.~5.7]{ANRR} to which we refer for the details. 
For the next statement, we define the operator $W_S \in \B(\Hrond)$ by
\begin{equation*}
W_S = 1+ \big\{\tfrac12\big(1-\tanh(2\pi A_+)-i\cosh(2\pi A_+)^{-1}\big)\otimes1_{\H}\big\} (S(L)-1).
\end{equation*}
Since $S(0) = \lim_{\lambda \to \infty}S(\lambda)=1$, the operator $W_S$ is also 
a Fredholm operator, as shown in \cite[Lem.~5.4]{ANRR}.
We can then deduce the following statement:

\begin{Lemma}\label{lem:equal-index}
Under the assumptions of Lemma \ref{lem:high-energy-limit},
the Fredholm operators $W_S$ and $W_\beta$ satisfy
$\ W_S W_{\beta} - W_{S\beta}  \in \K(\Hrond)\ $
and
$$
\index(W_S) = \index(W_{S\beta}).
$$ 
\end{Lemma}

\begin{proof}
The equality $W_S W_\beta = W_{S\beta}$ up to compact operators follows from the proof of \cite[Lem.~3.4]{Angus}, see also \cite[Lem.~5.8]{ANRR} for a similar statement.
The index claim follows from the fact that $\index(W_\beta) = 0$ and the composition rule for the Fredholm index.
\end{proof}

The next statement shows that the operator $\beta$ provides the correct regularisation for the operator $S$, and consequently $W_\beta$ will provide the correct regularisation
to the operator $W_S$.

\begin{Lemma}\label{lem:det-loop}
Under the assumptions of Lemma \ref{lem:high-energy-limit},
 the map $\lambda \mapsto \det\big(S(\lambda)\big)\det\big(\beta(\lambda)\big)$ satisfies
\begin{equation*}
\lim_{\lambda \searrow 0} \det\big(S(\lambda)\big)\det\big(\beta(\lambda)\big) = 1
\end{equation*}
and 
\begin{equation*}
\lim_{\lambda \to \infty} \det\big(S(\lambda)\big)\det\big(\beta(\lambda)\big) = 1.
\end{equation*}
\end{Lemma}

\begin{proof}
The result follows essentially from Lemma \ref{lem:high-energy-limit}. Using Equations \eqref{eq:trace-computation} and \eqref{eq:det-limits} we find
\begin{align*}
&\lim_{\lambda \to \infty} \det\big(S(\lambda)\big) \det\big(\beta(\lambda)\big) \\
&= \lim_{\lambda \to \infty} \frac{\det_2(1+uvR_0(\lambda-i0)v)}{\det_2(1+uvR_0(\lambda+i0)v)} \e^{-\Tr\big(uv(R_0(\lambda+i0)-R_0(\lambda-i0))v)} \e^{i \Tr(A(\lambda))} \\
&= 1.
\end{align*}
Since $S(0) = \beta(0) = 1$ we finally infer that $\det\big(S(0)\big) \det\big(\beta(0)\big) = 1$.
\end{proof}

We can thus state the main result of this section.

\begin{Proposition}\label{prop:comput}
Suppose that $V:\R^3\to \C_-$ is a bounded, measurable function satisfying \eqref{eq_decay_V} with $\alpha>7$ and $\Im V<0$ on a non-trivial open set.
Assume also that the operator $\overline{u} + vR_0(\lambda-i0)v$ is invertible in $\B(\H)$ for any $\lambda > 0$. 
Then the following equality holds: 
\begin{equation*}
\index(W_-) = \frac{1}{2\pi i} \int_0^\infty \d \lambda \;\!\left(\Tr\big(S(\lambda)^{-1}S'(\lambda)\big) + c \lambda^{-\frac12} \right)
\end{equation*}
where $c$ is defined in \eqref{eq_constant}.
\end{Proposition}

\begin{proof}
By Theorem \ref{thm_formula} and Lemma \ref{lem:equal-index} one has 
$$
\index(W_-)=
\index(W_S) = \index(W_{S\beta}).
$$ 
Note that by \cite[Thm.~3.5(a)]{simon} we have $\det\big(S(\lambda)\beta(\lambda)\big) = \det\big(S(\lambda)\big) \det\big(\beta(\lambda)\big)$. By Lemma \ref{lem:det-loop}, the map $\lambda \mapsto \det\big(S(\lambda)\beta(\lambda)\big)$ defines a loop, and using Gohberg-Kre\u{\i}n theory we can compute the index of $W_{S\beta}$ as
\begin{align*}
\index(W_{S\beta}) &= \wind\Big(\lambda\mapsto \det\big(S(\lambda)\beta(\lambda)\big)\Big) \\
&= \frac{1}{2\pi i} \int_0^\infty \d \lambda\;\! \frac{\frac{\d}{\d \lambda}\big[ \det\big(S(\lambda)\big) \det\big(\beta(\lambda)\big)\big]}{\det\big(S(\lambda)\big) \det\big(\beta(\lambda)\big)}  \\
&= \frac{1}{2\pi i} \int_0^\infty \d \lambda \;\!\Big( \Tr\big(S(\lambda)^{-1}S'(\lambda)\big) +i \frac{\d}{\d \lambda}\Tr\big(A(\lambda)\big) \Big).
\end{align*}
By Lemma \ref{lem:integrable}, $\int_0^\Lambda \, \d \lambda \left( \frac{\d}{\d \lambda} \ln\det\big(S(\lambda)\big) + c \lambda^{-\frac12} \right)$ is finite for all $\Lambda > 0$ and converges as $\Lambda \to \infty$. Noting also that 
\begin{align*}
\int_0^\infty\d \lambda\;\! \left( i \frac{\d}{\d \lambda}\Tr\big(A(\lambda)\big) - c \lambda^{-\frac12} \right)  &= 0,
\end{align*}
we have the desired result.
\end{proof}

To strengthen the previous result, we can use the description of $\Ran(W_-)$ given by Lemma \ref{lem:existence_wave}. Thus, coming back to the definition of the projection $P_-$ given in \eqref{eq_P_-} one has
\begin{equation*}
P_-\H = \H_\ads(H^*).
\end{equation*}
In addition, it is easily observed from the definition of these spaces in \eqref{eq:def_Hads}, by complex conjugation,
that 
$$
\dim\big(\H_\ads(H^*)\big) = \dim\big(\H_\ads(H)\big).
$$
Thus, collecting the content of Proposition \ref{prop:comput} together with \eqref{eq_P_-}, we finally deduce
our main result:

\begin{Theorem}\label{thm:analytic-formula}
Suppose that $V:\R^3\to \C_-$ is a bounded, measurable function satisfying \eqref{eq_decay_V} with $\alpha>7$ and $\Im V<0$ on a non-trivial open set.
Assume also that the operator $\overline{u} + vR_0(\lambda-i0)v$ is invertible in $\B(\H)$ for any $\lambda > 0$. Then one has
\begin{equation*}
\frac{1}{2\pi i} \int_0^\infty \d\lambda \;\! \left( \Tr\big(S(\lambda)^{-1}S'(\lambda)\big)+ c \lambda^{-\frac12} \right)   =  - \#\sigma_\ads(H) ,
\end{equation*}
where $c$ is defined in \eqref{eq_constant} and  $\#\sigma_\ads(H) = \dim\big(\H_\ads(H)\big)$.
\end{Theorem}

Recalling from Lemma \ref{lem:finite-dim} that $\H_\ads(H)=\H_\disc(H)$ if $\H_\ads(H)$ is finite-dimensional, we directly deduce from Theorem \ref{thm:analytic-formula} the following statement:

\begin{Corollary}
Under the conditions of Theorem \ref{thm:analytic-formula}, we have
\begin{equation*}
\H_\ads(H)=\H_\disc(H).
\end{equation*}
\end{Corollary}



\end{document}